\documentclass[onecolumn, draftcls,12pt]{IEEEtran}

\usepackage{amssymb,amsmath,amsfonts,amsthm}
\usepackage{algorithm,algorithmic}
\usepackage{boxedminipage}
\usepackage{cite}
\usepackage[dvips]{graphicx}
\usepackage{epsfig}
\usepackage{float}
\graphicspath{{figures/}}
\usepackage{graphicx}
\usepackage{graphics}
\usepackage{multirow}
\usepackage{threeparttable} 
\usepackage[usenames,dvipsnames]{color}
\usepackage{url,bm,xspace,dsfont}

\usepackage{verbatim}

\def\supp{\mathop{\mathrm{supp}}}  

\newtheorem{theorem}{Theorem}
\newtheorem{lemma}{Lemma}

\newtheorem{remark}{Remark}

\newtheorem{problem}{Problem}

\newcommand{\sr}{\stackrel}

\newcommand{\rar}{\rightarrow}

\newcommand{\tri}{\sr{\triangle}{=}}

\newcommand{\be}{\begin{equation}}
\newcommand{\ee}{\end{equation}}
\newcommand{\bea}{\begin{eqnarray}}
\newcommand{\eea}{\end{eqnarray}}
\newcommand{\bes}{\begin{eqnarray*}}
\newcommand{\ees}{\end{eqnarray*}}

\newcommand{\bi}{\begin{itemize}}
\newcommand{\ei}{\end{itemize}}
\newcommand{\ben}{\begin{enumerate}}
\newcommand{\een}{\end{enumerate}}

\newcommand{\bp}{\begin{problem}}
\newcommand{\ep}{\end{problem}}
\newcommand{\hso}{\hspace{.1in}}
\newcommand{\hst}{\hspace{.2in}}

\newcommand{\noi}{\noindent}

\newcommand{\bc}{\begin{center}}
\newcommand{\ec}{\end{center}}

\floatstyle{ruled}
\newfloat{algorithm}{H}{algo}
\floatname{algorithm}{\footnotesize Algorithm}

\begin{document}
%
\title{Variable Length Lossless Coding for Variational Distance Class: An Optimal Merging Algorithm }


\author{Themistoklis Charalambous, Charalambos D. Charalambous and Sergey Loyka
\thanks{T. Charalambous was with the Department of Electrical and Computer Engineering, University of Cyprus, Nicosia. Now he is with the Automatic Control Lab, Electrical Engineering Department and ACCESS Linnaeus Center, Royal Institute of Technology (KTH), Stockholm, Sweden.  Corresponding author's address: Osquldas v\"{a}g 10, 100-44 Stockholm, Sweden (E-mail: themisc@kth.se).}
\thanks{C.D. Charalambous is with the Department of Electrical and Computer Engineering, University of Cyprus, Nicosia 1678 (E-mail: chadcha@ucy.ac.cy).}
\thanks{Sergey Loyka is with the School of Information Technology and Engineering,
University of Ottawa, Ontario, Canada, K1N 6N5 (E-mail:  sergey.loyka@ieee.org).}
}

\maketitle

%
%
%
%
\begin{abstract}
In this paper we consider lossless source coding for a class of sources specified by the total variational distance ball centred at a fixed nominal probability distribution. The objective is to find a minimax average length source code, where the minimizers are the codeword lengths -- real numbers for arithmetic or Shannon codes -- while the maximizers are the  source distributions from  the total variational distance ball. Firstly, we examine the maximization of the average codeword length by converting it into an equivalent optimization problem, and we give the optimal codeword lenghts via a waterfilling solution. Secondly, we show that the equivalent optimization problem can  be solved via an optimal partition of the source alphabet, and re-normalization and merging of the fixed nominal probabilities. For the computation of the optimal codeword lengths we also develop a fast algorithm with a computational complexity of order ${\cal O}(n)$.  
\end{abstract}

\IEEEpeerreviewmaketitle

%
%
%
%
\section{Introduction}
\noi Lossless fixed to variable length source codes are often categorized into problems of known source  probability distribution and  unknown source probability distribution. For known source probability distribution several pay-offs are investigated in the literature, such as the average  codeword length \cite{2006:Cover}, the average redundancy of the codeword length \cite{2004:DrmotaSzpankowski}, the average of an exponential function of the codeword length \cite{1965:campbell_coding,1981:humblet_generalization,2008b:Baer}, and the average of an exponential function of the redundancy of the codeword length \cite{2006a:Baer,2008b:Baer}. Huffman type algorithms are also investigated for some of these pay-offs \cite{2006:Cover,2006a:Baer,2008b:Baer}. For the average codeword length pay-off the average redundancy is bounded below by zero and above by one. On the other hand,  if the true probability distribution of the source is unknown and the code is designed solely based on a given  nominal distribution (which is different than the true distribution), then the increase in the average codeword length due to incorrect knowledge of the true distribution is the relative  entropy between the true distribution and the nominal distribution \cite[\emph{Theorem 5.4.3}]{2006:Cover}. Such problems with unknown probability distribution are often investigated via universal coding and universal modeling, and the so-called Minimum Description Length (MDL) principle based on minimax techniques, by assuming the true source probability distribution belongs to a pre-specified  class of source distributions  \cite{1973:davisson,1980:davisson_Leon-Garcia, 2004:DrmotaSzpankowski,2003:JacquetSzpankowski,2007:FarzadCharalambous,2009:Gawrychowski_Gagie}, which may be parameterized or non-parameterized. Universal codes are often  examined under various pay-offs such as average minimax redundancy,  maximal minimax pointwise redundancy \cite{2004:DrmotaSzpankowski}, and variants of them involving the relative entropy between the true probability distribution and the nominal probability distribution \cite{2007:FarzadCharalambous,2009:Gawrychowski_Gagie}.

In this  paper, we investigate lossless variable length codes for a class of  source probability distributions  described by the total variational distance  ball, centred at a fixed (\'a priori) probability distribution (nominal), with the radius of the ball varying in the interval $[0,2]$. Since this problem falls into universal coding and modeling category  we formulate it using minimax techniques. The formal description of the coding problem which is made precise in the next section, is as follows. Given a class of source probability distributions described by the total variation metric centered at an \'a priori or nominal probability distribution ${\bm \mu} \in {\mathbb P}(\Sigma)$ (${\mathbb P}(\Sigma)$ the set of probability vectors on a finite alphabet set $\Sigma$) having radius $R\geq 0$ is defined by
\bea
{\mathbb B }_{{\bm \mu}} (R)   \tri \Big\{ {\bf \nu} \in {\mathbb P}({\Sigma}): ||{\bm \nu} -{\bm \mu}||_{TV} \tri  \sum_{x \in {\Sigma}} |\nu(x)-\mu(x)|\leq R \Big\}\; . \label{b32}
\eea
The pay-off may be anyone of those  mentioned earlier; we consider minimizing the  maximum of  the average codeword lengths defined by
\bea
{\mathbb L}_R({\bf l^\dagger}, {\bm \nu}) \tri    \max_{ {\bm \nu} \in  {\mathbb B}_{{\bm \mu}}({R})}  \sum_{ x \in {\Sigma}} l(x) \nu(x)\; .  \label{b33}
\eea
Specifically, our  main  objective is to find a prefix real-valued code  length vector ${\bf l}^\dagger$ which minimizes the pay-off ${\mathbb L}_{R}({\bf l}, {\bm \nu}^\dagger)$.

There are various reasons which motivated to consider the total variational distance class of sources  ${\mathbb B }_{{\bm \mu}} (R)$. Below, we describe some of these. Total variational distance can be used to  define the distance between the empirical distribution of a sequence and the fixed nonminal source distribution ${\bm \mu} \in {\mathbb P}(\Sigma)$ as follows. Given a sequence ${\bf x}^n \tri \{x_1, x_2, \ldots, x_n\} \in {\Sigma}^n$, let ${\nu}(x; {\bf x}^n)$ denote the empirical distribution of the sequence ${\bf x}^n$ defined by ${\bm \nu}(x; {\bf x}^n)\tri \frac{N(x| {\bf x}^n)}{n}$, with $N(x| {\bf x}^n)$ the number of occurence of $x$ in the sequence ${\bf x}^n$. For $\epsilon \geq 0$, we call a sequence ${\bf x}^n$ $\epsilon-$letter typical with respect to ${\bm \mu}$ if  $|{\nu}(x; {\bf x}^n)-\mu(x)| \leq \epsilon \mu(x), \forall x \in \Sigma$. The set of all such sequences ${\bf x}^n$ satisfying this inequality is called $\epsilon-$letter typical set $T_{\epsilon}^n({\bm \mu})$ with respect to $ {\bm {\mu}}$. Therefore, the total variational distance between the empirical distribution ${ \nu}(x; {\bf x}^n)$  and ${\bm \mu}$ satisfied the bound $||{\bm \nu}(\cdot; {\bf x}^n)-{\bm \mu}||_{TV} \leq \epsilon$. Therefore, the total variational ball radius can be easily obtained from observing specific sequences.  In this respect, ball  radius  $R$ is easily identified, and the  larger the value of $R$ the larger the admissible class of source distributions. The total variational distance is a true metric, hence it is a measure of difference between two distributions. 
By the properties of the distance metric then  $||{\bm \nu} -{\bm \mu}||_{TV} \leq ||{\bm \nu}||_{TV}  + ||{\bm \mu}||_{TV} =2$, hence $R$ is further restricted to the interval $[0,2]$. The two extreme cases are $R=0$ implying ${\bm \nu}= {\bm \mu}$, and $R=2$ implying that the support sets of ${\bm \nu}$ and ${\bm \mu}$ denoted by $\supp({\bm \nu})$ and   $\supp({\bm \mu})$, respectively, are non-overlapping, that is, $\supp({\bm \nu}) \cap \supp({\bm \mu}) = \emptyset$. Moreover, one of the most interesting properties of total variational distance ball is that any admissible ${\bm \nu } \in {\mathbb B }_{{\bm \mu}} (R)$ may not be absolutely continuous with respect to ${\bm \nu}$, denoted by ${\bm \nu} << {\bm \mu}$ and defined by $\mu(x) =0$ for some $x \in {\Sigma}$ then $\nu(x)=0$. Consequently, admissible distributions   ${\bm \nu} \in   {\mathbb B }_{{\bm \mu}} (R)$ can be defined on a larger alphabet than the nominal distribution ${\bm \mu}$, that is, the support set of ${\bm \mu}$ maybe a subset of  ${\Sigma}$.

There is an anthology of distances and distance metrics on the space of probability distributions which are related to total variational distance \cite{gibbs}, and therefore one can obtain various lower and upper bounds on the performance with respect to other classes of sources, based on (\ref{b33}). Consider for examples, the case when $ {\bm \nu}<< {\bm \mu}, \forall   {\bm \nu}\in  {\mathbb B }_{{\bm \mu}} (R)$; by   Pinsker's inequality \cite{1964:Pinsker},
\bes
 ||{\bm \nu} - {\bm \mu}||_{TV}^2 \leq
2 {\mathbb D}( {\bm \nu} || {\bm \mu} ), \hst \forall {\bm \nu} \in  {\mathbb B }_{{\bm \mu}} (R),  {\bm \nu} \in {\mathbb P}({\Sigma})
 \ees
 where ${\mathbb D}({\bm \nu}|| {\bm \mu} ) \tri \sum_{ x \in {\Sigma}} \nu(x) \log \frac{\nu(x)}{\mu(x)}$ denotes the Kullback-Leibler distance (or relative entropy distance) between ${\bm \nu}$ and ${\bm \mu}$. Thus, Pinsker's inequality implies that the total variational distance class is larger\footnote{The bound is tight in the sense that the ratio of ${\mathbb D}({\bm \nu}|| {\bm \mu} )$ and $||{\bm \nu}- {\bm \mu}||_{TV}$ can be arbitrarily close to $1/2$ \cite{1967:Csiszar,1969:Kemperman}.} than the class defined by replacing $||{\bm \nu}- {\bm \mu}||_{TV}$ by ${\mathbb D}({\bm \nu}|| {\bm \mu} )$.  Indeed it is  more appropriate especially when the probability distributions ${\bm \nu}$ and ${\bm \mu}$ are singular (resp. nearly singular) in which case ${\mathbb D}({\bm \mu}|| {\bm \mu} )=\infty$ (resp. very large), while $||{\bm \nu}- {\bm \mu}||_{TV} \leq 2$. 

The main contributions of this paper are the following.
\begin{enumerate}
\item The pay-off of maximizing the average codeword length  over the total variational distance ball is transformed into a new optimization problem which is convex with respect to the codeword length.
\item The problem can be solved by convex optimization tools and in a waterfilling-like fashion (see Theorem~\ref{theorem:waterfilling}), which requires numerical methods and no closed-form solution is provided. Note that this waterfilling structure does not belong to the family of watefilling solutions for which practical algorithms were proposed by Palomar \emph{et al.} \cite{2005:PalomarWaterfilling}.
\item The optimal code corresponding to the new optimization problem is then equivalent to a specific partition of the source alphabet, and re-normalization  and merging of entries of the initial source probability vector,  as a function of the radius of the ball $R \in [0,2]$, from which the optimal code is derived. An algorithm is presented which computes the weight vector ${\bm \nu}$, having a worst case computational complexity of order $\mathcal{O}(n)$. Our approach provides a methodology for the solution of such problems and also an approach for this new waterfilling structure.
\end{enumerate}

The paper is organized as follows.  In the next section, we formulate the minimax length problem and derive its equivalent optimization. In Section~\ref{sec:results}, we show that optimization Problem \ref{problem1} can be solved using convex optimization tools and a waterfilling approach. It is then transformed to an average coding problem (Problem \ref{problem3}), which is being solved via a fast algorithm that is based on re-normalization  of the initial source probabilities according to a merging  rule.  In Section~\ref{sec:examples}, illustrative examples demonstrate the validity of the proposed algorithm and provide better understanding on the impact of the distance parameter $R$ on the codeword lengths. The paper ends with the conclusions in Section \ref{sec:conclusions}.

%
%
%
%
\section{Problem Formulation}\label{sec:formulation}

\noi Consider a source generating outputs from a finite set of symbols, denoted by ${\Sigma } \tri \{x_1, x_2, \ldots, x_{ | {\Sigma }| }  \}$  of cardinality   $|{\Sigma}|$, according to a source probability distribution ${\bm \nu} \tri \{\nu(x): x \in {\Sigma}  \} \equiv \left(\nu(x_1), \nu(x_2), \ldots, \nu(x_{|{\Sigma}|})\right)$.  Source symbols are encoded into $D-$ary codewords (unless specified otherwise  $\log (\cdot) \tri \log_D (\cdot)$). A code ${\cal C} \tri \{c(x): x \in {\Sigma}\}$ for symbols in ${\Sigma}$ with image alphabet ${\cal D} \tri \{0, 1, 2, \ldots, D-1 \}$ is an injective map    $c: {\Sigma} \rar {\cal D}^*$, where ${\cal D}^*$ is the set of finite sequences drawn from ${\cal D}$.  For $x \in {\Sigma}$  each codeword $c(x) \in {\cal D}^*,  c \in {\cal C}$ is identified with a codeword length $l(x) \in {\mathbb Z}_+$, where ${\mathbb Z}_+$ is the set of non-negative integers. Thus, a  code ${\cal C}$ for source symbols from the alphabet ${\Sigma}$ is associated with the length function of the code $ l : {\Sigma} \rar {\mathbb Z}_+$, and a code defines a codeword length vector ${\bf l} \tri \{ l(x): x \in {\Sigma}\} \equiv \big(l(x_1), l(x_2), \ldots, l(x_{|{\Sigma}|})\big) \in {\mathbb Z}_+^{|{\Sigma}|}$. If, however, the integer constraint is relaxed by admitting real-valued length vectors ${\bf l} \in {\mathbb R}_+^{|{\Sigma}|}$, which satisfy the Kraft inequality (i.e., $\sum_{x\in {\Sigma}} D^{-l(x)}\leq 1$), then ${\cal L}\left( {\mathbb Z}_+^{|{\Sigma}|} \right)$ is replaced by
\bes
{\cal L}\left( {\mathbb R}_+^{|{\Sigma}|} \right)  \tri \Big\{ {\bf l} \in {\mathbb R}_+^{|{\Sigma}|} : \sum_{x \in {\Sigma}} D^{-l(x)} \leq 1 \Big\}.
\ees
Such codes give approximate solutions which are less computationally intensive \cite{2006:Cover}.

Suppose the source probability distribution ${\bm \nu}$ -- henceforth called {the true distribution} -- is unknown, while modeling techniques give  access to a nominal source probability distribution ${\bm \mu} \tri \{\mu(x): x \in {\Sigma}  \} \equiv \left(\mu(x_1), \mu(x_2), \ldots, \mu(x_{|{\Sigma}|})\right)$. Having constructed knowledge of the  nominal source distribution one may construct from empirical data via counting techniques, the  distance of the two distributions with respect to the total variation norm  $||{\bm \nu} - {\bm \mu}||_{TV}$. This will provide an estimate of the radius $R$, such that $||{\bm \nu} - {\bm \mu}||_{TV} \leq R$ and hence, characterize the set  ${\mathbb B }_{{\bm \mu}} (R)$ of all possible true distributions of the source.  Subsequently, the source coding problem for the class of sources ${\mathbb B }_{{\bf \mu}} (R)$ can be defined via minimax techniques as follows. Let ${\mathbb P}({\Sigma})$ denote the set of probability distributions on the alphabet $\Sigma$, and let ${\mathbb P}_{\bm \mu}({\Sigma})$
 denote the set of nominal probability distributions defined by
\begin{align*}
{\mathbb P}_{\bf \mu}({\Sigma})& \tri \Big\{ {\bm{\mu}} =\Big( \mu(x_1), \ldots, \mu(x_{|{\Sigma}|})\Big) \in {\mathbb R}_+^{|{\Sigma}|} : \\
&0 <\mu(x_i) \leq \mu(x_j), \forall i>j,  (x_i,x_j) \in {\Sigma},  \sum_{x \in {\Sigma}} \mu(x) =1 \Big\}.
\end{align*}

\noi The precise problem investigated is stated below.
\begin{problem}\label{problem1}
Given a fixed nominal distribution ${\bm \mu} \in   {\mathbb P}_{\bm \mu}({\Sigma})$ and distance parameter $R\in [0,2] $, define the class of source probability distributions by the total variational ball
\bea
{\mathbb B }_{{\bm \mu}} (R)   \tri \Big\{ {\bf \nu} \in {\mathbb P}({\Sigma}): ||{\bm \nu} -{\bm \mu}||_{TV} \tri  \sum_{x \in {\Sigma}} |\nu(x)-\mu(x)|\leq R \Big\} \label{b323}
\eea
and the average codeword length pay-off with respect to the true source probability distribution ${\bm \nu} \in {\mathbb B }_{{\bm \mu}} (R) \subset  {\mathbb P}({\Sigma})$ by
\bea
{\mathbb L}_R({\bf l}, {\bm \nu}) \tri \sum_{ x \in {\Sigma}} l(x) \nu(x)\; . \label{eq:01}
\eea
 The objective is to find a prefix code length vector ${\bf l}^\dagger \in {\mathbb R}_+^{|{\Sigma}|}$ (satisfying Kraft inequality), which minimizes the maximum average  codeword length pay-off defined by
 \bea
{\mathbb L}_R({\bf l}, {\bm \nu}^\dagger) \tri    \max_{ {\bm \nu} \in  {\mathbb B}_{{\bm \mu}}({R})}  \sum_{ x \in {\Sigma}} l(x) \nu(x) \; , \label{b333}
\eea
for all $R \in [0,2]$.
\end{problem}

\noi The characterization of optimal prefix code length vector ${\bf l}^\dagger \in {\mathbb R}_+^{|{\Sigma}|}$ is obtained by first converting  ${\mathbb L}_R({\bf l}, {\bm \nu}^\dagger)$ into an equivalent pay-off and then use the resulting pay-off to find the optimal code.

%
%
%
%
\section{Main results}\label{sec:results}

The objective of this section is twofold. First, to solve Problem~\ref{problem1} using an equivalent pay-off for which 
the optimal prefix code length vector ${\bf l}^\dagger \in {\mathbb R}_+^{|{\Sigma}|}$ is obtained using a waterfilling-like approach. Second, to find an explicit expression of the maximizing distribution ${\bm \nu} \in {\mathbb B }_{{\bm \mu}} (R)$. 
Subsequently, to derive certain properties of the maximizing distribution and identify how these properties are transformed into equivalent properties for the optimal codeword length vector. The main goal here is to identify how symbols are merged  together, and how the merging  changes as a function of the parameter $R \in [0,2]$, so that the optimal solution is characterized for all $R \in [0,2]$. From these properties the Shannon codeword lengths for Problem~\ref{problem1} will be found.

\subsection{Equivalent Pay-off and Waterfilling-Like Solution}
Let ${\mathbb M}_{sm}(\Sigma)$ denote the set of finite signed measures on $\Sigma$. Then, any ${\bm \eta} \in  {\mathbb M}_{sm}(\Sigma)$ has a Jordan decomposition $\big\{  {\bm \eta}^+, {\bm \eta}^-\big\}$ such that  ${\bm \eta }= {\bm \eta}^+ - {\bm \eta}^-$, and the total variation of ${\bm \eta}$ is defined by $|| {\bm \eta}||_{TV} \tri {\bm \eta}^+(\Sigma) + {\bm \eta}^-(\Sigma)$.  Define the following subset ${\mathbb M}_0(\Sigma) \tri \Big\{ {\bm \eta} \in {\mathbb M}_{sm}(\Sigma): {\bm \eta}(\Sigma)=0\Big\} \subset {\bm \eta} \in  {\mathbb M}_{sm}(\Sigma)$. 
For ${\bm \xi} \in {\mathbb M}_0(\Sigma)$, then $\xi(\Sigma)=0$, which implies that ${\bm \xi}^+(\Sigma)= {\bm \xi}^-(\Sigma)$, and hence ${\bm \xi}^+(\Sigma)= {\bm \xi}^-(\Sigma)=\frac{||{\bm  \xi}||_{TV}}{2}$. Define ${\bm \xi}  \tri {\bm \nu} -{\bm \mu} \in  {\mathbb M}_0(\Sigma)$. Since ${\bf l} \in {\mathbb R}_+^{|{\Sigma}|}$ are non-negative the following inequalities are obtained.

\begin{align}
 \sum_{ x \in {\Sigma}} l(x) \nu(x)  &= \sum_{ x \in {\Sigma}} l(x) \xi(x) + \sum_{ x \in {\Sigma}} l(x) \mu(x) \nonumber  \\
&=\sum_{ x \in \Sigma} l(x) \left(\xi^{+}(x)- \xi^{-}(x)\right) + \sum_{ x \in {\Sigma}} l(x) \mu(x) \nonumber   \\
&=\sum_{ x \in \Sigma} l(x) \xi^{+}(x)- \sum_{ x \in \Sigma} l(x) \xi^{-}(x)+  \sum_{ x \in {\Sigma}} l(x) \mu(x) \nonumber  \\
& \leq \max_{x\in \Sigma} l(x){\bm \xi}^{+}(\Sigma) - \min_{x\in \Sigma} l(x) {\bm \xi}^-(\Sigma) + \sum_{ x \in {\Sigma}} l(x) \mu(x)  \nonumber  \\
& = \max_{x\in \Sigma} l(x)\frac{||{\bm \xi} ||_{TV}}{2} - \min_{x\in \Sigma} l(x)\frac{||{\bm \xi} ||_{TV}}{2} + \sum_{ x \in {\Sigma}} l(x) \mu(x) \nonumber   \\
&=\Big\{ \max_{x\in \Sigma} l(x) - \min_{x\in \Sigma} l(x) \Big\}\frac{||{\bm \xi} ||_{TV}}{2} + \sum_{ x \in {\Sigma}} l(x) \mu(x) \label{f2}
\end{align}
For a given ${\bm \mu} \in   {\mathbb P}_{\bm \mu}(\Sigma)$ define  the set $\widetilde{{\mathbb B }}_{{\bm \mu}} (R)$ by 
\bea
\widetilde{{\mathbb B }}_{{\bm \mu}} (R) \tri \Big\{ {\bm \xi} \in {\mathbb M}_0(\Sigma): {\bm \xi} = {\bm \nu}- {\bm \mu},  \hso {\bm \nu } \in {\mathbb P}(\Sigma), \hso ||{\bm \xi}|| \leq R\Big\}. \label{sb1}
\eea
For any $ \xi \in \widetilde{\mathbb B}_{{\bm \mu}}(\Sigma)$ then ${\bm \xi} = ({\bm \nu}-{\bm \mu})^+ -({\bm \nu}- {\bm \mu})^- \equiv {\bm \xi}^+- {\bm \xi}^-$.\\
Moreover,  the upper bound in the right hand side of (\ref{f2}) is achieved by   ${\bm \xi}^\dagger \in \widetilde{{\mathbb B }}_{{\bm \mu}} (R)$ as follows. Let
\bes
x^0 \in \Sigma^0 & \tri & \Big\{ x \in \Sigma: l(x) = \max \{l(x): x\in \Sigma\} \equiv l_{\max} \Big\}, \\
  x_0 \in \Sigma_0 & \tri & \Big\{ x \in \Sigma: l(x) = \min\{l(x): x\in \Sigma\} \equiv l_{\min} \Big\}.
\ees
Take
\bea
{\bm \xi}^\dagger(x) ={\bm \nu}^\dagger(x)-{\bm \mu}(x) = \frac{R}{2} \Big(\delta_{x^0}(x) - \delta_{x_0}(x)\Big), \hst x \in \Sigma \label{nm1}
\eea
where $\delta_y(x)$ denotes the point mass distribution concentrated at $y\in \Sigma$. This is indeed a signed measure with total variation $||{\bm \nu}^\dagger-{\bm \mu}||_{TV}=R$, and $\sum_{\Sigma} l(x) ({\nu}^\dagger- { \mu})(x)=\frac{R}{2}\Big(l_{\max}-l_{\min}\Big)$. \\
Hence, by using (\ref{nm1}) as a candidate of the maximizing distribution then
\bea \sum_{\Sigma} l(x) { \nu}^\dagger(x) =  \frac{R}{2}  \Big\{ \max_{x \in {\Sigma}} l(x)  -  \min_{x \in \Sigma} l(x) \Big\}    + \sum_{x \in \Sigma} l(x) \mu(x), \label{f2n}
\eea
where ${\bm \xi}^\dagger$  satisfies the constraint $||{\bm \xi}^\dagger||_{TV}= ||{\bm \nu}^\dagger-{\bm \mu}||_{TV} =R$.\\

Thus, $ {\mathbb L}_R({\bf l}, {\bm \nu}^\dagger)$ in \eqref{b333} is equivalent to pay-off \eqref{f2n}. At this stage it is clear that Problem~\ref{problem1} is equivalent to minimizing \eqref{f2n} subject to the Kraft inequality. This problem can be solved by a wide variety of convex optimization methods; in the following theorem we provide a waterfilling-like solution obtained by the Karush-Kuhn-Tucker theorem. 
Before we proceed further we discuss some generalizations.

\begin{remark}
\label{gen}
The derivations  leading to (\ref{f2n}) is generic in the sense that it is an optimization of a linear functional over the total variational ball, and hence it is applicable to a variety of problems. Below, we discuss two generalizations. \\
\noi 1) Theorem~\ref{theorem:waterfilling} holds for countable alphabets $\Sigma$ since the derivations do not depend on any assumption on  the cardinality of $\Sigma$.\\
\noi 2) The derivation leading to (\ref{f2n}) holds for abstract alphabets, such as complete separable metric spaces $(\Sigma, d)$ with ${\cal B}(
\Sigma)$ the $\sigma-$algebra of Borel sets in $\Sigma$ with the following modifications. $\nu, \mu$ are probability measures on $\Sigma$, ${ l}$ is a non-negative bounded continuous function ${ l} : \Sigma \rar [0,\infty)$, $\sum_{x \in \Sigma} l(x) \nu(x), \sum_{x \in \Sigma} l(x) \mu(x)$ are replaced by integrals   $\int_{x \in \Sigma} l(x) \nu(dx), \int_{x \in \Sigma} l(x) \mu(dx)$, and the $\min, \max$ operations   are replaced by $\sup, \inf$ operations (unless $\Sigma$ is compact). In this case, 
For any $l$ which is bounded continuous and non-negative, from (\ref{f2n}) we have:

\bea \int_{\Sigma} l(x) \nu^\dagger(dx) & =&  \frac{R}{2}  \Big\{ \sup_{x \in {\Sigma}} l(x)  -  \inf_{x \in \Sigma} l(x) \Big\}    +  \int_{\Sigma} l(x) \mu(dx) 
   \label{f32nn}
\eea
and  
\begin{align}
\int_{\Sigma^0} \nu^\dagger(dx) = \mu(\Sigma^0) + \frac{R}{2} \in [0,1],   & \hso  \int_{\Sigma_0} \nu^\dagger(dx) = \mu(\Sigma_0) - \frac{R}{2} \in [0,1], \nonumber \\
& \nu^\dagger(A)= \mu(A), \hso \forall A \subseteq \Sigma \setminus \Sigma^0\cup \Sigma_0 \label{cs}
\end{align}
Moreover, even in this abstract case, the first right hand side term of (\ref{f32nn}) is related to the oscillator semi-norm of $l$ by
\bea
osc(l)\tri \sup_{(x,y) \in \Sigma \times \Sigma}|l(x)-l(y)| =2 \inf_{\alpha \in {\mathbb R}} || l-\alpha||_{\infty} =  \sup_{x \in {\Sigma}} l(x)  -  \inf_{x \in \Sigma} l(x)  \label{osc2}
\eea
Although, generalization 2) is not pursued in this paper, one can infer that the generic result is of interest for classes of distributions on abstract alphabets.  
\end{remark}

\begin{theorem}\label{theorem:waterfilling}
Consider pay-off ${\mathbb L}_{R}({\bf l}, {\bm \nu})$  and real-valued prefix codes. Let $\underline{w}$ and $\overline{w}$ such that 
\begin{align}\label{waterf1}
\sum_{x \in {\Sigma}} \big(\underline{w}-\mu(x)\big)^+=\frac{R}{2},
\end{align}
\noi and
\begin{align}\label{waterf2}
\sum_{x \in {\Sigma}} \big(\mu(x) - \overline{w}\big)^+=\frac{R}{2},
\end{align}
where $(f)^+ = \max(0,f)$ and $R \in [0,2]$. The distribution ${\bm \nu}^\dagger \in {\mathbb B }_{{\bm \mu}} (R)$ which minimizes the maximum average  codeword length pay-off ${\mathbb L}_R({\bf l}, {\bm \nu}^\dagger)$ for all $R \in [0,2]$ is given by
\begin{align}
\nu^\dagger(x)= 
\begin{cases} 
\overline{w} & \text{if $\mu(x)>\overline{w}$,} \\
\mu(x) & \text{if $\underline{w} \leq \mu(x)\leq \overline{w}$,} \\
\underline{w} &\text{if $\mu(x) < \underline{w}$.}
\end{cases}
\end{align}

\end{theorem}

\begin{proof}
See Appendix \ref{waterfilling}.
\end{proof}

An example of the solution to the coding problem with real valued prefix codes for a total variational distance ball is obtained from Theorem~\ref{theorem:waterfilling} and it is depicted in Figure~\ref{waterproofing}.
\begin{figure}[H]
\centering
\includegraphics[width=0.7\columnwidth]{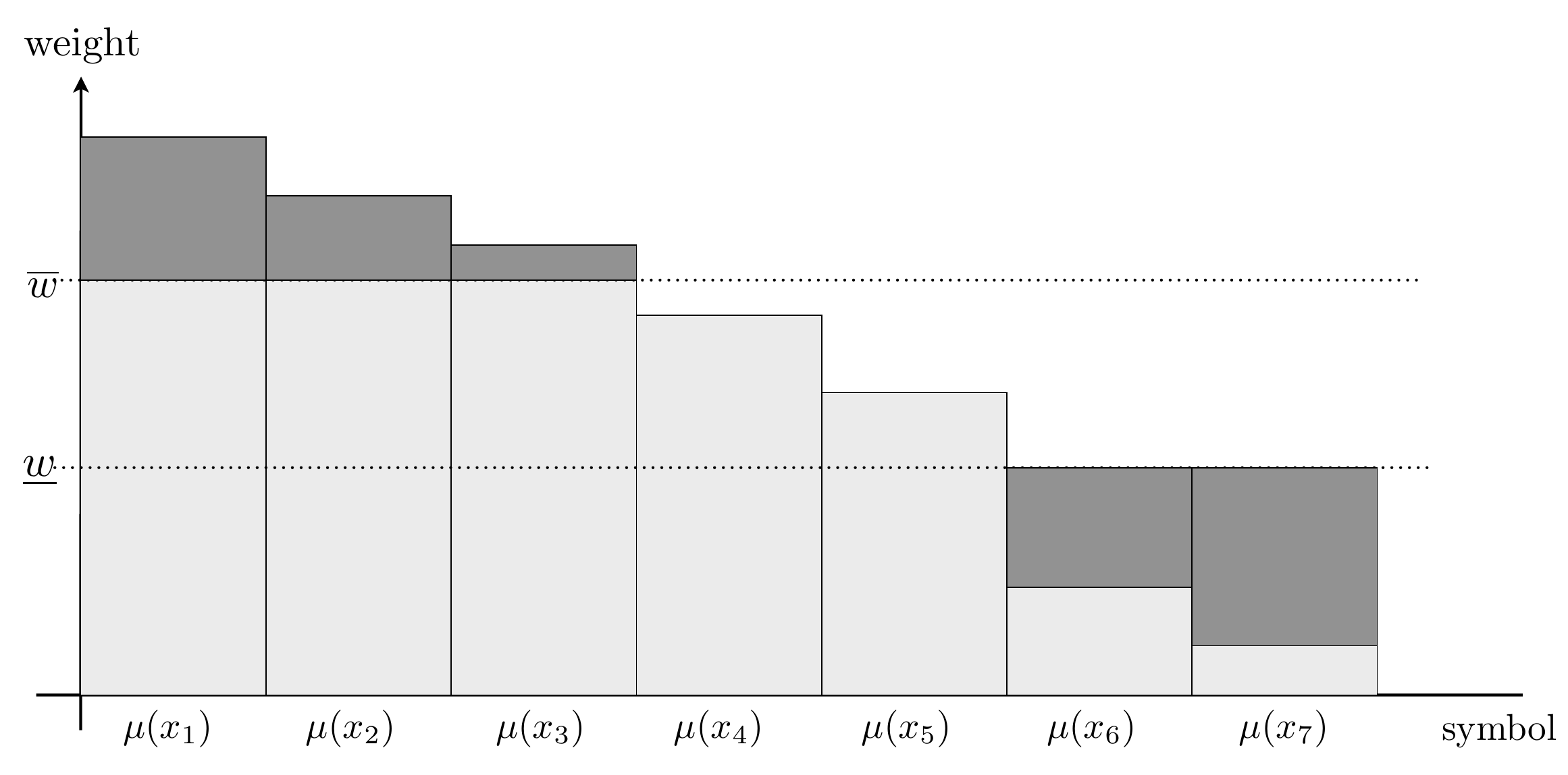}
\caption{Example demonstrating the solution of the coding problem using a watefilling-like fashion. In the example of the figure, $ {\bm \nu}^\dagger=\{\overline{w}, \overline{w}, \overline{w}, \mu(x_4), \mu(x_5) ,\underline{w} ,\underline{w}  \}$. }\label{waterproofing}
\end{figure}
 
A similar problem is considered in \cite{2010:interplay}, where the Shannon entropy of an unknown distribution is maximized subject to a variational distance constraint between a nominal distribution and the unknown distribution.  With completely different approach, \cite{2010:interplay} are able to provide a similar solution to the waterfilling approach described in this section, which however cannot incorporate classes of sources on abstract alphabets.

\subsection{Optimal Weights and Merging Rule}

\noi The pay-off ${\mathbb L}_R({\bf l}, {\bm \nu}^\dagger)$ can be written as
\begin{align}
 {\mathbb L}_R({\bf l}, {\bm \nu}^\dagger)
 = \sum_{ x \in {\Sigma\setminus\Sigma^o \cup \Sigma_o}} l(x) \mu(x)+\left( \sum_{ x \in \Sigma^o}  \mu(x)+\frac{R}{2} \right) l_{\max} +
\left( \sum_{ x \in \Sigma_o}  \mu(x)-\frac{R}{2} \right) l_{\min},  \label{n2}
 \end{align}
where
\begin{align*}
& \sum_{x \in \Sigma^0} \nu^\dagger(x) = \sum_{ x \in \Sigma^o} \mu(x) + \frac{R}{2} \in [0,1], \hso  \sum_{x \in \Sigma_o} \nu^\dagger(x) = \sum_{ x \in \Sigma_o} \mu(x) - \frac{R}{2}\in [0,1], \\
& \nu^\dagger(x)= \mu(x), \hso \forall x \in \Sigma \setminus \Sigma^o\cup \Sigma_o, \hso 0\leq \nu^\dagger(x)\leq 1, \hso \forall x \in \Sigma .
\end{align*}
The above expression makes the dependence on the disjoint sets $\Sigma^o$, $\Sigma_o$ and  $\Sigma\setminus\Sigma^o \cup \Sigma_o$ explicit. The sets  remain to be identified so that a solution to the coding problem exists for all $R\in [0,2] $. Note that $l_{\min}$,  $l_{\max}$ and sets  $\Sigma^o$ and $\Sigma_o$ depend parametrically on $R\in [0,2] $. This explicit dependence will often be omitted for simplicity of notation.

Define $\alpha \equiv R/2$, then Problem \ref{problem1} becomes equivalent to Problem \ref{problem3}, stated below.
\begin{problem}\label{problem3}
Given a fixed nominal distribution ${\bm \mu} \in   {\mathbb P}_{\bm \mu}({\Sigma})$ and distance parameter $\alpha \in [0,1]$, define the pay-off as follows:
\bea
{\mathbb L}_\alpha({\bf l}, {\bm \mu}) \tri \sum_{ x \in {\Sigma\setminus\Sigma^o \cup \Sigma_o}} l(x) \mu(x)+\left( \sum_{ x \in \Sigma^o}  \mu(x)+\alpha \right) l_{\max} +
\left( \sum_{ x \in \Sigma_o}  \mu(x)-\alpha \right) l_{\min}
, \label{eq:03}
\eea
 The objective is to find a prefix code length vector ${\bf l}^\dagger \in {\mathbb R}_+^{|{\Sigma}|}$ which minimizes the pay-off ${\mathbb L}_\alpha({\bf l}, {\bm \mu}) $,
for all $\alpha \in [0,1]$ such that the Kraft inequality holds; i.e., $\sum_{x\in {\Sigma}} D^{-l(x)}\leq 1$.
\end{problem}

In this section, the optimal real-valued prefix codeword lengths vector ${\bf l}^\dagger$  minimizing pay-off ${\mathbb L}_{\alpha}({\bf l}, {\bm \mu})$ as a function of $\alpha \in [0,1]$ and the initial source probability vector ${\bm \mu}$, are recursively calculated via re-normalization and merging. For any specific $\hat{\alpha} \in [0,1]$, a fast algorithm (of linear complexity in the worst case) is devised which obtains the optimal real-valued prefix codeword lengths minimizing pay-off ${\mathbb L}_{\hat{\alpha}}({\bf l}, {\bm \mu})$.

\noi Define
\begin{subequations}\label{n3}
\begin{align}
 \sum_{ x \in \Sigma^o}   \nu_{\alpha}(x) &= \sum_{ x \in \Sigma^o}  \mu(x)+\alpha \in [0,1],  \label{n3a} \\
 \sum_{ x \in \Sigma_o}  \nu_{\alpha}(x) &= \sum_{ x \in \Sigma_o}  \mu(x)-\alpha \in [0,1],  \label{n3b} \\
\nu_{\alpha}(x) &= \mu(x),~ x \in  {\Sigma\setminus\Sigma^o \cup \Sigma_o}.  \label{n3c}
 \end{align}
\end{subequations}
Using \eqref{eq:03} and \eqref{n3} the pay-off ${\mathbb L}_\alpha({\bf l}, {\bm \mu})$ is written as a function of the new weight vector as follows.
\bea
{\mathbb L}_\alpha({\bf l}, {\bm \mu})\equiv {\mathbb L}({\bf l}, {\bm \nu}_\alpha) \tri \sum_{x \in { \Sigma}} \nu_\alpha(x) l(x), \hst \alpha \in [0,1]. \label{n4}
\eea
The new weight vector ${\bm \nu}_{\alpha}$ is a function of $\alpha$ and the source probability vector ${\bm \mu} \in {\mathbb P}_{\bm \mu}({\Sigma})$, and it is defined over the three disjoint sets  $\Sigma^o$, $\Sigma_o$ and  $\Sigma\setminus\Sigma^o \cup \Sigma_o$. It can be easily verified that $0\leq \nu_{\alpha}(x) \leq 1,~\forall x \in \{\Sigma^o,~\Sigma_o\}$ (if any of the weights was negative, then someone could easily choose a very large $l(x)$ and the pay-off ${\mathbb L}_\alpha({\bf l}, {\bm \mu})\equiv {\mathbb L}({\bf l}, {\bm \nu}_\alpha)$ would be negative) and $\sum_{ x \in \Sigma }\nu_{\alpha}(x) =1,  \forall \alpha \in [0,1]$.

\begin{lemma}\label{lemma_optLengths}
 The real-valued prefix codes minimizing pay-off ${\mathbb L}_\alpha({\bf l}, \bm{\mu})$ for $\alpha \in [0,1]$ are given by
\bea \label{eq:solution1}
l^\dagger(x) = \left\{ \begin{array}{ll}
-\log{\Big(\mu(x)\Big)} &  x \in \Sigma\setminus\Sigma_o\cup\Sigma^o  \\
-\log{\left( \frac{ \sum_{ x \in \Sigma^o}  \mu(x)+\alpha }{|\Sigma^o|} \right)}, &  x \in \Sigma^o \\
-\log{\left( \frac{ \sum_{ x \in \Sigma_o}  \mu(x)-\alpha }{|\Sigma_o|} \right)}, &  x \in \Sigma_o
\end{array} \right.
\eea
where $\Sigma_o$ and $\Sigma^o$  remain to be specified.
\end{lemma}

\begin{proof}
See Appendix \ref{proof_lemma_optLengths}.
\end{proof}

\noi The point to be made regarding Lemma~\ref{lemma_optLengths} is twofold: (a) since for $\alpha \in [0,1]$ the pay-off ${\mathbb L}_\alpha({\bf l}, \bm{\mu})$ is continuous in ${\bf l}$ and the constraint set defined by Kraft inequality is closed and bounded and hence compact, an optimal code length vector ${\bf l}^\dagger$ exists, and (b) the optimal code is given by (\ref{eq:solution1}).

From the characterization of optimal code length vector of Lemma~\ref{lemma_optLengths}, it follows that ${\mathbb L}_\alpha({\bf l}^\dagger, {\bm \mu})= -\sum_{x \in \Sigma} \nu_\alpha(x) \log \nu_\alpha^\dagger(x)  \geq {\mathbb H}({\bm \nu}_\alpha)$, where $ {\mathbb H}({\bm \nu}_\alpha)$ denotes  the entropy of the probability distribution ${\bm \mu}$. Equality holds if, and only if, $\nu_\alpha(x)=\nu_\alpha^\dagger(x), \forall x \in \Sigma$. Therefore, for $\alpha \in [0,1]$ the weights satisfying (\ref{n3}) and corresponding to the optimal code length vector are uniquely represented via ${\bm \nu}_\alpha={\bm \nu}_\alpha^\dagger$. Further, by rounding up the optimal codeword lengths (i.e., $l^\ddagger(x) \tri \lceil -\log \nu_\alpha^\dagger(x)  \rceil$) Kraft inequality remains valid and hence ${\mathbb H}({\bm \nu}_\alpha) \leq \sum_{x \in \Sigma} l^\ddagger (x) \nu_\alpha(x) < {\mathbb H}({\bm \nu}_\alpha) +1$.

\noi The next lemma describes monotonicity properties of the weight vector ${\bm \nu}_\alpha$ as a function of the probability vector ${\bm \mu}$, for all $\alpha \in [0,1]$.

\begin{lemma}
\label{lemma_weights1}
Consider pay-off ${\mathbb L}_\alpha({\bf l}, \bm{\mu})$ and real-valued prefix codes. The following hold:
\begin{enumerate}
\item For $\{x,y\}  \subset  \Sigma$, if $\mu(x)\leq \mu(y)$ then $\nu_\alpha(x) \leq \nu_\alpha(y)$, for all $\alpha \in [0,1]$. Equivalently,  $\mu(x_1) \geq \mu(x_2) \geq \ldots \geq \mu(x_{|\Sigma|})>0$ implies $\nu_\alpha(x_1) \geq \nu_\alpha(x_2) \geq \ldots \geq \nu_\alpha(x_{|\Sigma|})>0$, for all $\alpha \in [0,1]$.
\item For $y \in \Sigma\setminus\Sigma_o\cup\Sigma^o $,  $\nu_\alpha(y)$ is constant and independent of $\alpha \in [0,1]$.
\item For $x \in \Sigma^o $, $\nu_\alpha(x)$ is a monotonically increasing  function of $\alpha \in [0,1]$.
\item For $x \in \Sigma_o $, $\nu_\alpha(x)$ is a monotonically decreasing  function of $\alpha \in [0,1]$.
 \end{enumerate}
\end{lemma}

\begin{proof}
See Appendix \ref{proof_lemma_weights1}.
\end{proof}

Next, the merging rule which described how the weight vector $\bm{\nu}_\alpha$ changes as a function of  $\alpha \in [0,1]$ is identified, such that a solution to the coding problem is completely characterized  for arbitrary cardinalities $|\Sigma^o|$ and $|\Sigma_o|$, and not necessarily distinct probabilities, for  any $\alpha \in [0,1]$. Clearly, there  is a  minimum $\alpha$ called $\alpha_{\max}$ such that  for any $\alpha \in [\alpha_{\max}, 1]$ there is no compression. 

Consider the complete characterization of the solution, as $\alpha$ ranges over $[0,1]$, for any initial probability vector ${\bm \mu}$ (not necessarily consisting of distinct entries).  Then, $|\Sigma_o|+|\Sigma^o| \in \{1, 2, \ldots, |\Sigma|-1\}$ while for $|\Sigma_o|+|\Sigma^o| =|\Sigma|$, $\alpha\in[\alpha_{\max},1]$, there is no compression since the weights are all equal. 

\noi Define
\begin{align*}\label{bk}
\beta_{k_1} & \tri \min\left\{\beta \in [0,1]: \nu_\beta(x_{|\Sigma|-(k_1-1)})= \nu_\beta(x_{|\Sigma|-k_1})\right\}, \hst    k_1 \in \{1,\ldots, |\Sigma|-1\},  \hst \beta_0 \tri 0,  \\
\gamma_{k_2} & \tri \min\left\{\gamma \in [0,1]: \nu_\gamma(x_{(k_2-1)})= \nu_\gamma(x_{k_2})\right\}, \hst    k_2 \in \{2,\ldots, |\Sigma|-1\},  \hst \gamma_0 \tri 0, \\
\alpha_k  & \tri \max\left\{\beta_{k_1},\gamma_{k_2}\right\},\hst k=k_1+k_2, \hst \alpha_0 \tri 0.
\end{align*}

\noi By Lemma \ref{lemma_weights1} the weights are ordered, hence  $\alpha_1$ is the smallest value of $\alpha \in [0,1]$ for which two weights become equal; this can occur because the two smallest weights become equal ($\beta_1<\gamma_1$), or because the two biggest weights become equal ($\gamma_1<\beta_1$).

\noi Since for $k=0$, $\nu_{\alpha_0}(x)=\nu_{0}(x)=\mu(x), \forall x  \in \Sigma$, is the set of initial symbol probabilities, let $\Sigma^{o,0}$ denote the singleton set $\{x_{|\Sigma|} \}$ and $\Sigma_{o,0} $ denote the singleton set $\{x_{1} \}$. Specifically,
\begin{align}
\Sigma^{o,0} &\tri \left\{x\in \{x_{|\Sigma|} \}:  \mu^{\flat}\tri \min_{ x\in\Sigma} \mu(x)= \mu(x_{|\Sigma|})  \right\} , \\
\Sigma_{o,0} &\tri \left\{x\in \{x_{|\Sigma|} \}:  \mu^{\sharp}\tri \max_{ x\in\Sigma} \mu(x)= \mu(x_{1})  \right\} .
\end{align}
Similarly, $\Sigma^{o,1}$ is defined as the set of symbols in $\{x_{|\Sigma|-1}, x_{|\Sigma|}\}$ whose weight evaluated at $\beta_1$ is equal to the minimum weight  $\nu_{\beta_1}^\flat$ and $\Sigma_{o,1}$ is defined as the set of symbols in $\{x_{1}, x_{2}\}$ whose weight evaluated at $\gamma_1$ is equal to the maximum weight  $\nu_{\gamma_1}^\sharp$:
\begin{align}
\Sigma^{o,1} & \tri \Big\{ x \in \{ x_{|\Sigma|-1}, x_{|\Sigma|} \}:  \nu_{\beta_1}(x)= \nu_{\beta_1}^\flat \Big\}, \\
\Sigma_{o,1} & \tri \Big\{ x \in \{ x_{1}, x_{2} \}:  \nu_{\gamma_1}(x)= \nu_{\gamma_1}^\sharp \Big\}.
\end{align}
\noi In general, for  a given value of $\alpha_k, k \in \{1,\ldots, |\Sigma|-1\}$, define
\begin{align}
\Sigma^{o,k_1} & \tri \Big\{ x \in \{ x_{|\Sigma|-k_1-1}, x_{|\Sigma|-k_1}, \ldots, x_{|\Sigma|} \}:  \nu_{\beta_{k_1}}(x)= \nu_{\beta_{k_1}}^\flat \Big\}, \\
\Sigma_{o,k_2} & \tri \Big\{ x \in \{x_{1} \ldots, x_{k_2}, x_{k_2+1} \}:  \nu_{\gamma_{k_2}}(x)= \nu_{\gamma_{k_2}}^\sharp \Big\}.
\end{align}
and for $k=k_1+k_2$,  $\alpha_k = \max \left\{{\beta_{k_1},\gamma_{k_2}}\right\}$.

\begin{lemma}\label{prop1}
Consider  pay-off ${\mathbb L}_\alpha({\bf l}, {\bm \mu})$ and real-valued prefix codes. For $k_1,k_2 \in \{0, 1, 2, \ldots, |\Sigma|-1\}$, then
\begin{align}
\nu_\beta(x_{ |\Sigma|-k_1}) = \nu_{\beta}(x_{|\Sigma|})=\nu_{\beta}^\flat, \hst \beta \in [\beta_{k_1}, \beta_{k_1+1}) \subset [0,1), \\
\nu_\gamma(x_{k_2}) = \nu_{\gamma}(x_{1})=\nu_{\gamma}^\sharp, \hst \gamma \in [\gamma_{k_2}, \gamma_{k_2+1}) \subset [0,1).
\end{align}
Further, the cardinality of sets $\Sigma^{o,k_1}$ and $\Sigma_{o,k_2}$ is $(k_1+1)$ and $(k_2+1)$, respectively.
\end{lemma}

\begin{proof}
See Appendix \ref{proof_prop1}.
\end{proof}

The  next theorem describes how the weight vector ${\bm \nu}_\alpha$ changes as a function of $\alpha \in [0,1]$ so that the solution of the coding problem can be characterized.

\begin{theorem}\label{main_theorem}
Consider pay-off ${\mathbb L}_\alpha({\bf l}, {\bm \mu})$  and real-valued prefix codes.
For  $\alpha \in [\alpha_k,\alpha_{k+1})$, $k\in\{0,1,\ldots,|\Sigma|-1\}$, the optimal weights
\begin{align*}
{\bm{\nu}^{\dagger}_{\alpha}} \tri \{{\nu}^{\dagger}_{\alpha}(x): x \in\Sigma  \} \equiv \big({\nu}^{\dagger}_{\alpha}(x_1), {\nu}^{\dagger}_{\alpha}(x_2), \ldots, {\nu}^{\dagger}_{\alpha}(x_{|\Sigma|})\big),
\end{align*}
 are given by
\begin{align}
\nu^{\dagger}_{\alpha}(x) =
\begin{cases}
\mu(x),~ &x \in   {\Sigma\setminus\Sigma^{o} \cup \Sigma_o}, \\ 
\displaystyle \frac{\sum_{x\in\Sigma^{o,k_1}}\mu (x)+\alpha}{1+k_1} , ~&x \in \Sigma^{o,k_1}, \\
\displaystyle \frac{\sum_{x\in\Sigma_{o,k_2}}\mu (x)-\alpha}{1+k_2} , ~&x \in \Sigma_{o,k_2},
\end{cases}
\label{weights_update}
\end{align}
where
\begin{align}
\displaystyle \beta_{k_1+1} &= (k_1+1) \mu(x_{|{\Sigma}|-(k_1+1)}) -\sum_{x\in \Sigma^{o,k_1}}\mu(x),  \label{betakplus1} \\
\displaystyle \gamma_{k_2+1} &=\sum_{x\in \Sigma_{o,k_2+1}}\mu(x) - (k_2+1)\mu(x_{k_2+1}),  \label{gammakplus1} \\
\alpha_{k+1}&=\min{\{\beta_{k_1+1},  \gamma_{k_2+1}\}}. \label{akplus1}
\end{align}
Moreover, the minimum $\alpha$, called $\alpha_{\max}$, such that for  $\alpha \in [\alpha_{\max},1] $  there is no compression, is given by
\begin{align}
\alpha_{\max}=(k_1^*+1)\frac{1}{|\Sigma|}-\sum_{x\in \Sigma^{o,k_1^*}}\mu(x),
\end{align}
where $k_1^*$ is the number of probabilities $\mu(x)\in \Sigma$ that are less than $1/|\Sigma|$.
\end{theorem}

\begin{proof}
\noi The derivation of Theorem~\ref{main_theorem} is based on the Lemmas introduced prior to Theorem~\ref{main_theorem}. 
By Lemma \ref{prop1}, for $\alpha \in [\alpha_k, \alpha_{k+1})$, the lowest probabilities that are equal, change together forming a total weight given by
\begin{align*}
\sum_{x \in\Sigma^{o,k_1}}\nu_{\alpha}(x)&=|\Sigma^{o,k_1}| \nu_{\alpha}^\flat = \sum_{ x \in \Sigma^{o,k_1}}  \mu(x)+\alpha ,
\end{align*}
whereas the highest probabilities that are equal, change together forming a total weight given by
\begin{align*}
\sum_{x \in \Sigma_{o,k_2}}\nu_{\alpha}(x)&=|\Sigma_{o,k_2}| \nu_{\alpha}^\sharp =  \sum_{ x \in \Sigma_{o,k_2}}  \mu(x)-\alpha.
\end{align*}
At $\alpha=\beta_{k_1 +1}$, each weight is equal to $\mu(x_{|\Sigma|-(k_1+1)})$ and from Lemma \ref{prop1} we have
\begin{align*}
\mu(x_{|\Sigma|-(k_1+1)})= \sum_{ x \in \Sigma^{o,k_1}}  \mu(x)+\beta_{k_1+1} \Rightarrow \beta_{k_1+1}= (k_1+1) \mu(x_{|{\Sigma}|-(k_1+1)}) -\sum_{x\in \Sigma^{o,k_1}}\mu(x).
\end{align*}
Similarly, it is shown for $\alpha=\gamma_{k_2 +1}$ that
\begin{align*}
\gamma_{k_2+1} =\sum_{x\in \Sigma_{o,k_2+1}}\mu(x) - (k_2+1)\mu(x_{k_2+1}).
\end{align*}
Once we find $ \beta_{k_1+1}$ and $\gamma_{k_2+1} $, $\alpha_{k+1}$ will denote the value of $\alpha$ for which there is merging and this will be the smallest between $ \beta_{k_1+1}$ and $\gamma_{k_2+1} $.
The minimum $\alpha$, called $\alpha_{\max}$, such that for  $\alpha \in [\alpha_{\max},1] $  there is no compression, is obtained when all the weights converge to the average probability, i.e. $\nu^{\dagger}_{\alpha}=1/|\Sigma|$. We know that this probability will lie between two nominal probabilities whose weights will converge one from above and one from below. Hence, we can easily find the maximum cardinalities of $\Sigma^{o,k_1}$ and $\Sigma_{o,k_2}$. Once, the cardinality is known we can use one of the equations for finding $ \beta_{k_1+1}$ and $\gamma_{k_2+1} $ to find $\alpha_{\max}$. Here, we use \eqref{betakplus1} and $\alpha_{\max}$ can be expressed as follows:
\begin{align}
\alpha_{\max}=(k_1^*+1)\frac{1}{|\Sigma|}-\sum_{x\in \Sigma^{o,k_1^*}}\mu(x) \in [0,1].
\end{align}
\end{proof}

\noi Theorem~\ref{main_theorem} facilitates the computation of the optimal real-valued prefix codeword lengths vector ${\bf l}^\dagger$  minimizing pay-off ${\mathbb L}_{\alpha}({\bf l}, {\bm \mu})$ as a function of $\alpha \in [0,1]$ and the initial source probability vector ${\bm \mu}$, via re-normalization and merging. Specifically, the optimal weights are found recursively calculating $\beta_{k_1}, k_1 \in \{0,1,\ldots, |\Sigma|-1\}$ and $\gamma_{k_2}, k_2 \in \{0,1,\ldots, |\Sigma|-1\}$ and hence $\alpha_k, k \in \{0,1,\ldots, |\Sigma|-1\}$.  For any specific $\hat{\alpha} \in [0,1]$ an algorithm is given next, which describes  how to obtain the optimal real-valued prefix codeword lengths minimizing pay-off ${\mathbb L}_{\hat{\alpha}}({\bf l}, {\bm \mu})$.

The main difference between the solutions emerging from Theorems~\ref{theorem:waterfilling}~and~\ref{main_theorem} is the following. Theorem~\ref{theorem:waterfilling} simplifies the problem and complexity by boiling the problem down to the numerical solution of a waterfilling equation, while Theorem~\ref{main_theorem} finds an explicit expression of the weights. While both approaches solve the problem, Theorem~\ref{main_theorem} finds an explicit expression, thus revealing several properties of the solution and the impact on $\alpha$ on the optimal real-valued prefix codeword lengths. 

\subsection{An Algorithm for Computing the Optimal Weights}

\noi For any probability distribution  ${\bm \mu} \in {\mathbb P}({\Sigma})$ and $\alpha \in [0,1]$ an algorithm is presented to compute the optimal weight vector ${\bm \nu}_\alpha$ of Theorem~\ref{main_theorem}. By Theorem~\ref{main_theorem} (see also Fig.~\ref{probs1} for a schematic representation of the weights for different values of $\alpha$), the weight vector ${\bf \nu}_\alpha$ changes piecewise linearly as a function of  $\alpha \in [0,1]$.

\begin{figure}[H]
\centering
\includegraphics[width=\columnwidth]{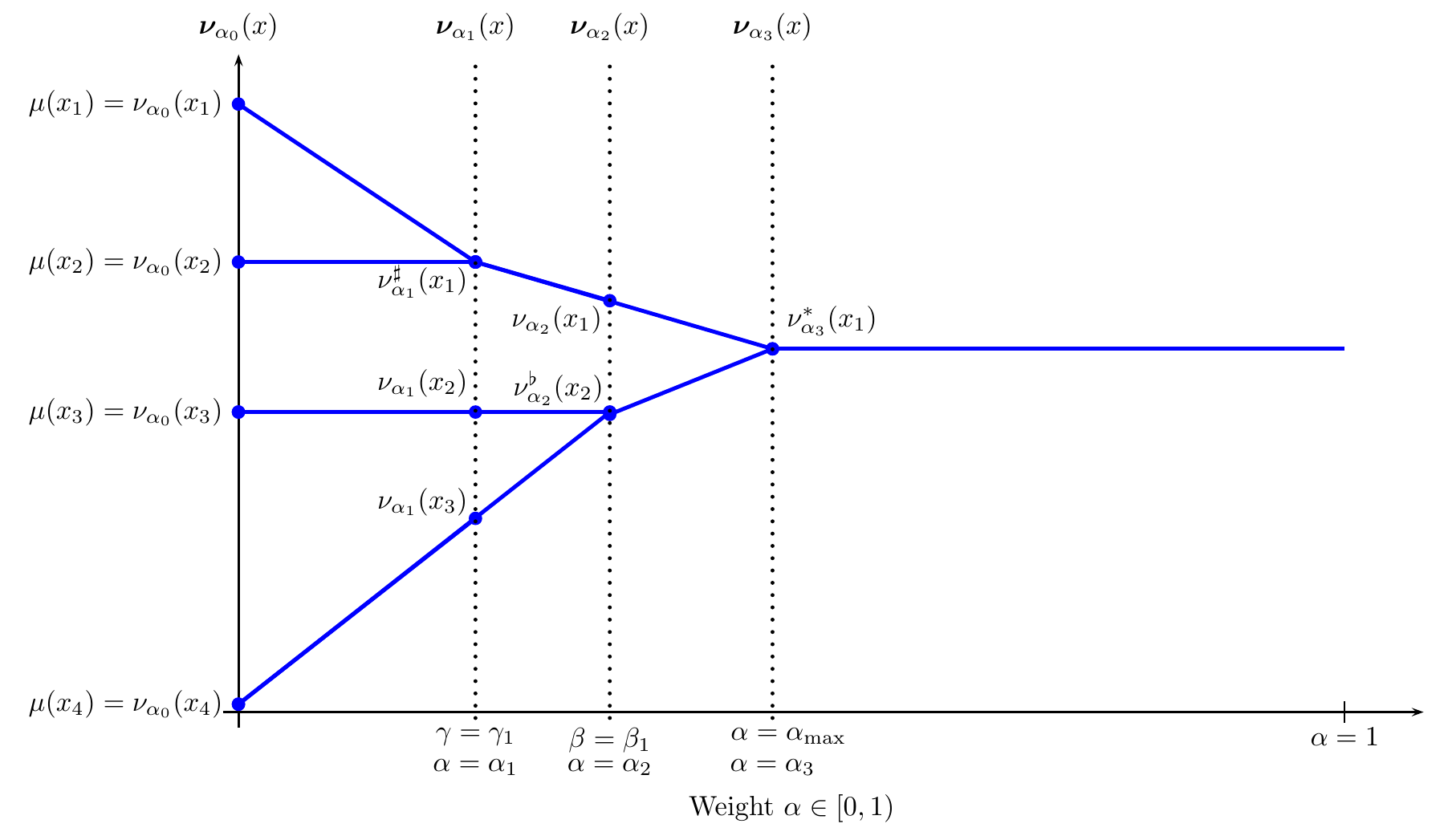}
\caption{A schematic representation of the weights for different values of $\alpha$. The weight vector ${\bf \nu}_\alpha$ changes piecewise linearly as a function of  $\alpha \in [0,1]$.}\label{probs1}
\end{figure}

Given a  specific value of $\hat{\alpha} \in [0,1]$, in order to calculate the weights $\nu_{\hat{\alpha}}(x)$, it is sufficient  to determine  the values of $\alpha$ at the intersections by using \eqref{akplus1}, up to the value of $\alpha$ for which the intersection gives a value greater than $\hat{\alpha}$, or up to the last intersection (if all the intersections give a smaller value of $\alpha$) at $\alpha_{\max}$ beyond which there is no compression. For example, if $\alpha_1<\hat{\alpha}<\alpha_2$,  find all $\alpha$'s at the intersections up to and including $\alpha_2$ and subsequently, the  weights at $\hat{\alpha}$ can be found by using \eqref{weights_update}. Specifically,  check first if $\hat{\alpha}\geq \alpha_{\max}$. If yes, then the weights are equal to $1/|\mathcal{X}|$. If  $\hat{\alpha}< \alpha_{\max}$, then find $ \alpha_1,   \ldots, \alpha_m$, $m\in \mathbb{N}$, $m\geq 1$, until   ${\alpha_{m-1}}<\hat{\alpha}\leq {\alpha_{m}}$. As soon as the $\alpha$'s at the intersections are found, the  weights at $\hat{\alpha}$ can be found by using \eqref{weights_update}. The algorithm is easy to implement and extremely fast due to its low computational complexity. The worst case scenario appears when $\alpha_{|\mathcal{X}|-2}<\hat{\alpha}< \alpha_{\max}=\alpha_{|\mathcal{X}|-1}$, in which all $\alpha$'s at the intersections are required to be found. In general, the worst case complexity of the algorithm is $\mathcal{O}(n)$. The complete algorithm is depicted under Algorithm \ref{charalambousTV_algorithm}.

{\footnotesize
\begin{algorithm}
\caption{\small Algorithm for Computing the Weight Vector ${\bm \nu}_\alpha$ }
\label{charalambousTV_algorithm}
\begin{algorithmic}
\STATE $\,$\\
\STATE \textbf{initialize}
\STATE {$\quad\, \bm{\mu}=\left(\mu(x_1), \mu(x_2), \ldots, \mu(x_{|{\Sigma}|})\right)^T$, $\alpha = \frac{R}{2}$ $\quad\, k=0$, $k_1=0$, $k_2=0$, $\quad \beta_0=0$  $\gamma_0 = 0$}
\WHILE{$\displaystyle \alpha_{k}< \frac{R}{2}$}
\STATE {$\quad\, \displaystyle \beta_{k_1+1}= (k_1+1) \mu(x_{|{\Sigma}|-(k_1+1)}) -\sum_{x\in \Sigma^{o,k_1}}\mu(x)$, $\quad\, \displaystyle \gamma_{k_2+1}=\sum_{x\in \Sigma_{o,k_2}}\mu(x) - (k_2+1)\mu(x_{k_2+1})$}
\IF{$\beta_{k_1+1}<\gamma_{k_2+1}$}
\STATE {$\alpha_{k+1}=\beta_{k_1+1}$, $\quad\, k \leftarrow k + 1$, $\quad\, k_1 \leftarrow k_1 + 1$}
\ELSIF{$\beta_{k_1+1}>\gamma_{k_2+1}$}
\STATE {$\alpha_{k+1}=\gamma_{k_2+1}$, $\quad\, k \leftarrow k + 1$, $\quad\, k_2 \leftarrow k_2 + 1$}
\ELSIF{$\beta_{k_1+1}=\gamma_{k_2+1}$}
\STATE {$\alpha_{k+1}=\beta_{k_1+1}$, $\alpha_{k+2}=\gamma_{k_2+1}$,$\quad\, k \leftarrow k + 2$, $k_1 \leftarrow k_1 + 1$, $k_2 \leftarrow k_2 + 1$}
\ENDIF
\ENDWHILE
\IF{$\alpha_{k}=\beta_{k_1}$}
\STATE {$k_1 \leftarrow k_1 - 1$}
\ELSIF{$\alpha_{k}=\gamma_{k_2}$}
\STATE {$k_2 \leftarrow k_2 - 1$}
\ELSE
\STATE {$k_1 \leftarrow k_1 - 1$, $k_2 \leftarrow k_2 - 1$}
\ENDIF

\FOR{$n =1$ to $k_2+1$}
\STATE {$\displaystyle \nu^{\dagger}_{\frac{R}{2}}(x_{n})=\frac{\sum_{x\in\Sigma_{o,k_2}}\mu (x)-\frac{R}{2}}{1+k_2}$, $n \leftarrow n+ 1$}
\ENDFOR
\FOR{$n =k_2+2$ to $|{\Sigma}|-k_1-1$}
\STATE {$\displaystyle\nu^{\dagger}_{\frac{R}{2}}(x_{n})=\mu(x_n)$, $n \leftarrow n+ 1$}
\ENDFOR
\FOR{$n =|{\Sigma}|-k_1$ to $|{\Sigma}|$}
\STATE {$\displaystyle\nu^{\dagger}_{\frac{R}{2}}(x_{n})=\frac{\sum_{x\in\Sigma^{o,k_1}}\mu (x)+\frac{R}{2}}{1+k_1} $, $n \leftarrow n+ 1$}
\ENDFOR

\RETURN $\bm{\nu}^{\dagger}_{\frac{R}{2}}$.
\end{algorithmic}
\end{algorithm}}

%
%
%
%
\section{Illustrative Examples}\label{sec:examples}

\noi This section presents illustrative examples of the optimal codes derived in this paper.
\subsection{Illustrative theoretical example}

The following example is introduced to illustrate how the weights $\bm{\nu}_\alpha$ and the cardinality of the sets $\Sigma_o$ and  $\Sigma^o$ change as a function of  $\alpha \in [0,1]$.

Consider the special case when the probability vector $\bm{\mu}(x) \in {\mathbb P}(\Sigma)$ consists of distinct probabilities, e.g.,   that $\mu(x_{|\Sigma|})<\mu(x_{|\Sigma|-1})$ and $\mu(x_{2})<\mu(x_{1})$. The goal is to characterize the weights in a subset of $\alpha\in [0,1]$, such that  $\nu_\alpha(x_{|\Sigma|})<\nu_\alpha(x_{|\Sigma|-1})$ and $\nu_\alpha(x_{2})<\nu_\alpha(x_{1})$ hold. Since $\Sigma^o=\{x_{|\Sigma|}\}$ $(|\Sigma^o|=1)$ and $\Sigma_o=\{x_{1}\}$ $(|\Sigma_o|=1)$ then
\begin{align*}
{\mathbb L}_\alpha({\bf l}, \bm{\mu}) =\Big(\mu(x_{|\Sigma|})+\alpha \Big) l_{\max}+ \Big(\mu(x_{1})-\alpha \Big) l_{\min}+    \sum_{x   \in   \Sigma\setminus\Sigma_o\cup\Sigma^o}\mu(x) l(x) = \sum_{x \in \Sigma} l(x) \nu_\alpha(x) .
\end{align*}
\noi where the weights are given by $\nu_{\alpha}(x) = \mu(x),~x \in  \Sigma\setminus\Sigma_o\cup\Sigma^o$, $\nu_{\alpha}(x_{|\Sigma|}) = \mu(x_{|\Sigma|})+\alpha$ and $\nu_{\alpha}(x_{1}) = \mu(x_{1})-\alpha$ (by Lemma~\ref{lemma_weights1}). For any $\alpha \in [0,1]$ such that the condition $\nu_\alpha(x_{|\Sigma|}) < \nu_\alpha(x_{|\Sigma|-1})$ and $\nu_\alpha(x_{2})<\nu_\alpha(x_{1})$ hold, the optimal codeword lengths are given by $-\log \nu_\alpha(x), x \in \Sigma$, and  this region of $\alpha \in [0,1]$ for which $|\Sigma^o|=1$ and $|\Sigma^o|=1$ satisfies the following inequalities
\begin{align}
\mu(x_{|\Sigma|})+\alpha<\mu(x_{|\Sigma|-1}) \quad \text{and} \quad \mu(x_{1})-\alpha>\mu(x_{2})
\end{align}
Equivalently,
\begin{align*}
\left\{\alpha \in [0,1]: \alpha<  \min\{ \mu(x_{|\Sigma|-1})-\mu(x_{|\Sigma|}), \mu(x_{1})-\mu(x_{2})\}  \right\}.
\end{align*}
Hence, under the conditions $\Sigma^o=\{x_{|\Sigma|}\}$ $(|\Sigma^o|=1)$ and $\Sigma_o=\{x_{1}\}$ $(|\Sigma_o|=1)$, the optimal codeword lengths are given by $-\log \nu_\alpha(x), x \in \Sigma$ for   $\alpha < \alpha_{1} \tri  \min\{ \mu(x_{|\Sigma|-1})-\mu(x_{|\Sigma|}), \mu(x_{1})-\mu(x_{2})\} $, while for     $\alpha \geq \alpha_{1}$ the form of the minimization problem changes, as more weights $\nu_{\alpha}(x)$ enter either $\Sigma^o$ or $\Sigma_o$, and the cardinality of that set is changed; that is, the partition of $\Sigma$ into $ \Sigma\setminus\Sigma_o\cup\Sigma^o$, $\Sigma^o$ and $\Sigma_o$ is changed. Note that when $\mu(x_{|\Sigma|}) = \mu(x_{|\Sigma|-1})$, in view of the continuity of the weights $\bm{\nu}_\alpha$ as a function of $\alpha \in [0,1]$, the above optimal codeword lengths are only characterized for the singleton point $\alpha=\alpha_1={0}$, giving the classical codeword lengths. For $\alpha \in (0,1)$ the problem should be reformulated.

Without loss of generality, and for the sake of simplicity of exposition of this example, suppose that $\mu(x_{1})-\mu(x_{2})< \mu(x_{|\Sigma|-1})-\mu(x_{|\Sigma|})$.  If we now consider the case for which $\alpha > \alpha_{1}$ and  $|\Sigma_o|=2$ the problem can be written as
\begin{align*}
{\mathbb L}_\alpha({\bf l}, \bm{\mu}) =\Big(\mu(x_{|\Sigma|})+\alpha \Big) l_{\max}+ \Big(\mu(x_{1})+\mu(x_{2})-\alpha \Big) l_{\min}+    \sum_{x   \in   \Sigma\setminus\Sigma_o\cup\Sigma^o}\mu(x) l(x) = \sum_{x \in \Sigma} l(x) \nu_\alpha(x) .
\end{align*}
For any $\alpha \in [\alpha_1,1)$ such that the conditions $\nu_\alpha(x_{|\Sigma|}) < \nu_\alpha(x_{|\Sigma|-1})$ and $\nu_\alpha(x_{3}) < \nu_\alpha(x_{2})$ hold, the optimal codeword lengths are given by $-\log \nu_\alpha(x), x \in \Sigma$ and this region is specified  by
\begin{align}\label{alpha1}
 \left\{ \alpha \in [0,1]: \alpha_{1} <\alpha < \min\{ \mu(x_{|\Sigma|-1})-\mu(x_{|\Sigma|}), \mu(x_{1})+\mu(x_{2})-2\mu(x_{3})\} \right\}.
\end{align}
The procedure is repeated and the problem is reformulated until all $\nu_{\alpha}(x)= \mu(x),~x \in  \Sigma\setminus\Sigma_o\cup\Sigma^o$ join the sets $\Sigma^o$ and $\Sigma_o$. Eventually, for large $\alpha$ sets $\Sigma^o$ and $\Sigma_o$ will merge together and $l(x)=l_{\min}=l_{\max}$.

\subsection{Optimal weights for all $\alpha \in [0,1]$ for specific probability distributions}

\noi Consider binary codewords and a source with $|{\Sigma}|=4$ and probability distribution
\begin{align*}
\displaystyle \bm{\mu}=\left(\begin{array}{cccc}
   \frac{8}{15} &  \frac{4}{15} &  \frac{2}{15} &  \frac{1}{15}
\end{array}\right).
\end{align*}
Using Algorithm \ref{charalambousTV_algorithm} one can find the optimal weight vector $\bm{v}_\alpha^{\dagger}$ for different values of $\alpha \in [0,1]$ for which pay-off \eqref{eq:03} of Problem~\ref{problem3} is minimized. The weights for all $\alpha \in [0,1]$ can be calculated iteratively by calculating $\alpha_k$ for all $k\in \{ 0, 1, 2, 3\}$ and noting that the weights vary linearly with $\alpha$ (Figure \ref{ex1_p}).

\begin{figure}[H]
\centering
\includegraphics[width=0.55\columnwidth]{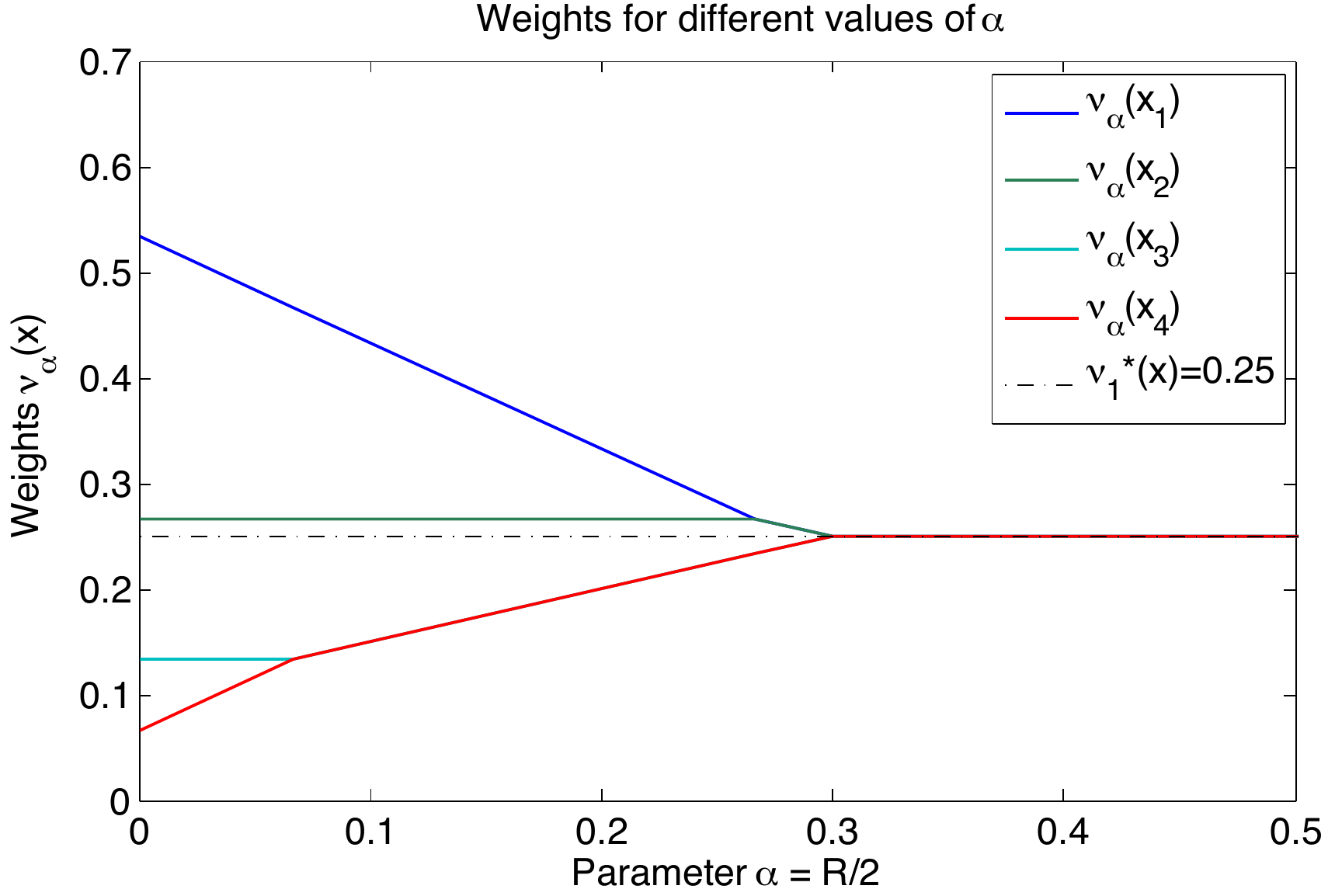}
\caption{A schematic representation of the weights for different values of $\alpha$ when ${\bm \mu}=(\frac{8}{15},\frac{4}{15},\frac{2}{15},\frac{1}{15})$.}\label{ex1_p}
\end{figure}

\noi The first merging occurs when
\begin{align}
\alpha_{1}= \min\{ \mu(x_{|\Sigma|-1})-\mu(x_{|\Sigma|}), \mu(x_{1})-\mu(x_{2})\} =\min \left\{\frac{2}{15}-\frac{1}{15},  \frac{8}{15}- \frac{4}{15}\right\}=\min\left\{\frac{1}{15}, \frac{4}{15}\right\}.
\end{align}
For $\alpha=\alpha_1$ the optimal weights according to are given by ${\bm \nu}_{\alpha_1}=(\frac{7}{15},\frac{4}{15},\frac{2}{15},\frac{2}{15})$.

Now consider binary codewords and a source with $|{\Sigma}|=5$ and probability distribution
\begin{align*}
\displaystyle \bm{\mu}=\left(\begin{array}{ccccc}
  \frac{16}{31} & \frac{8}{31} &  \frac{4}{31} &  \frac{2}{31} &  \frac{1}{31}
\end{array}\right).
\end{align*}
Using Algorithm \ref{charalambousTV_algorithm} one can find the optimal weight vector $\bm{v}_\alpha^{\dagger}$ for different values of $\alpha \in [0,1]$ for which pay-off \eqref{eq:03} of Problem~\ref{problem3} is minimized.

\begin{figure}[H]
\centering
\includegraphics[width=0.55\columnwidth]{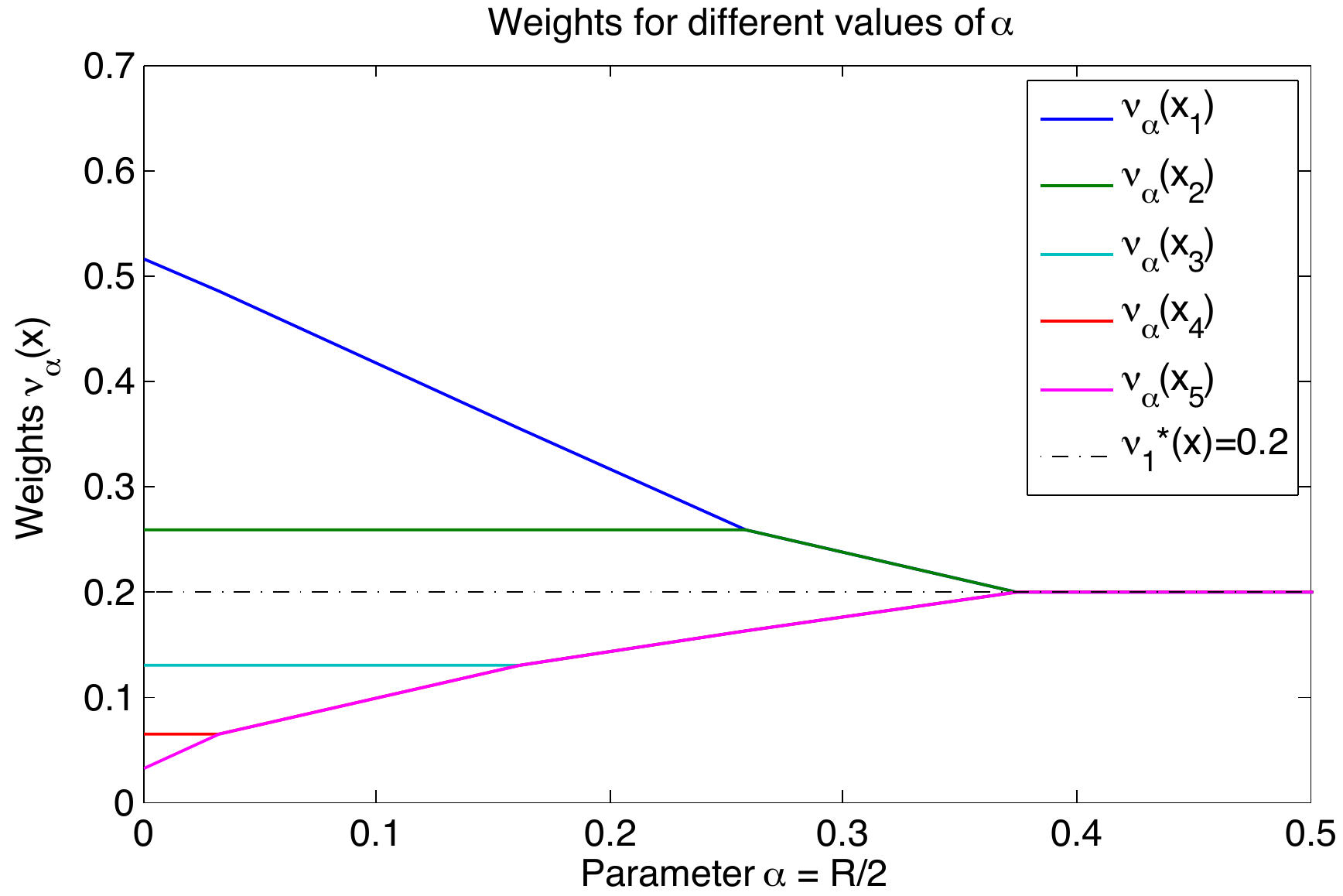}
\caption{A schematic representation of the weights for different values of $\alpha$ when $\mu=(\frac{16}{31}, \frac{8}{31},\frac{4}{31},\frac{2}{31},\frac{1}{31})$.}\label{ex2_p}
\end{figure}

Given the weights, we transformed the problem into a standard average length  coding problem, in which  the optimal codeword lengths can be easily calculated for all $\alpha$'s and they are equal to $ \lceil-\log(\nu_\alpha(x))\rceil, \forall x \in \Sigma$.

%
%
%
%
\section{Conclusions}\label{sec:conclusions}

The solution to a minimax average codeword length lossless coding problem for the class of sources described by the total variational ball is presented. First, the problem is transformed into an optimization one by finding the expresion of the maximization over the total variational ball. Subsequently, we give two solutions to the initial minimax coding problem for the class of sources. The first solution is given in terms  a waterfilling with two distinct levels. The second solution is given by a  procedure based on re-normalization of the  fixed nominal source probabilities according to a specific merging  rule of symbols. Several properties of the solution are introduced and an algorithm is presented which computes the minimax codeword lengths.  Illustrative examples corroborating the performance of the codes are presented.

Although, we consider the average codeword length, other pay offs can be considered, such as, average redundancy, average of exponential function of  the redundancy, pointwise redundancy etc., without much variation in the method of solution.

%
%
%
%

%
%
%
%
\appendices

\section{Proofs}

\subsection{Proof of Theorem~\ref{theorem:waterfilling}}\label{waterfilling}

\noi The problem can be expressed as 
\bea
\max_{s}\min_{t}\min_{l} \Big\{ \alpha (t-s) + \sum_{ x \in {\Sigma}} l(x) \mu(x)\Big\},  \quad \forall x\in \Sigma, \label{b31}
\eea
subject to the Kraft inequality and the constraints $l(x)\leq t$ $\forall x\in {\Sigma}$ and $l(x)\geq s$, $\forall x\in {\Sigma}$. 
By introducing real-valued  Lagrange multipliers $\lambda(x)$ associated with the constraint $l(x)\leq t$, $\forall x\in {\Sigma}$, $\sigma(x)$ associated with the constraint $l(x)\geq s$, $\forall x\in {\Sigma}$, and a real-valued Lagrange multiplier $\tau$ associate with the Kraft inequality, the augmented pay-off is defined by
\begin{align*} 
 {\mathbb L}_\alpha({\bf l}, {\bf p}, {\bf \lambda}, {\bf \sigma}, {\tau}) & \tri \alpha (t-s) + \sum_{ x \in {\Sigma}} l(x) \mu(x) +  \tau\left(\sum_{x \in {\Sigma}}D^{-l(x)}-1\right) \\
 &+\sum_{x \in {\Sigma}}\lambda(x) (l(x) - t) +\sum_{x \in {\Sigma}}\sigma(x) (s - l(x)) \; .
\end{align*}
The augmented pay-off is  a convex and differentiable function with respect to ${\bf l}$, $t$ and $s$. Denote the real-valued minimization over ${\bf l}, t, s, {\bf \lambda}, {\bf \sigma}, \tau$ by  ${\bf l}^\dagger, t^\dagger, s^\dagger, {\bf \lambda}^\dagger$,  ${\bf \sigma}^\dagger$ and $\tau^\dagger$. By the Karush-Kuhn-Tucker theorem, the following conditions are necessary and sufficient for  optimality.

\begin{eqnarray}
\frac{\partial }{\partial  l(x)}  {\mathbb L}_\alpha({\bf l}, {\bf \mu}, t, s, {\bf \lambda}, {\bf \sigma}, {\tau}) \arrowvert_{ {\bf l}={\bf l}^\dagger, {\bf \lambda}={\bf \lambda}^\dagger, t=t^\dagger, s=s^\dagger, {\bf \sigma}={\bf \sigma}^\dagger, \tau=\tau^\dagger} &=&0, \forall x\in {\Sigma} \label{lg1} \\
\frac{\partial }{\partial  t}  {\mathbb L}_\alpha({\bf l}, {\bf \mu}, t, s, {\bf \lambda}, {\bf \sigma}, {\tau}) \arrowvert_{ {\bf l}={\bf l}^\dagger, {\bf \lambda}={\bf \lambda}^\dagger, t=t^\dagger,  s=s^\dagger, {\bf \sigma}={\bf \sigma}^\dagger, \tau=\tau^\dagger}&=&0,  \label{lg1a} \\
\frac{\partial }{\partial  s}  {\mathbb L}_\alpha({\bf l}, {\bf \mu}, t, s, {\bf \lambda}, {\bf \sigma}, {\tau}) \arrowvert_{ {\bf l}={\bf l}^\dagger, {\bf \lambda}={\bf \lambda}^\dagger, t=t^\dagger,  s=s^\dagger, {\bf \sigma}={\bf \sigma}^\dagger, \tau=\tau^\dagger}&=&0,  \label{lg1b} \\
\sum_{x \in {\Sigma}}D^{-l^\dagger(x)}  -1\ &\leq& 0,  \label{lg2} \\
\tau^\dagger \cdot  \left(\sum_{x \in {\Sigma}}D^{-l^\dagger(x)}  -1\right)          &=&0, \label{lg3} \\
\tau^\dagger &\geq& 0, \label{lg4} \\
 l^\dagger(x)-t^\dagger &\leq& 0, \forall x\in {\Sigma},  \label{lg5} \\
\lambda^\dagger(x) \cdot  \left( l^\dagger(x)-t^\dagger \right)          &=&0, \forall x\in {\Sigma}, \label{lg6} \\
\lambda^\dagger(x) &\geq& 0, \forall x\in {\Sigma}. \label{lg7}\\
s^\dagger - l^\dagger(x) &\leq& 0, \forall x\in {\Sigma},  \label{lg5a} \\
\sigma^\dagger(x) \cdot  \left(s^\dagger- l^\dagger(x) \right)          &=&0, \forall x\in {\Sigma}, \label{lg6a} \\
\sigma^\dagger(x) &\geq& 0, \forall x\in {\Sigma}. \label{lg7a}
\end{eqnarray}
Differentiating with respect to ${\bf l}$, the following equation is obtained:
\begin{align}
\frac{\partial        }{\partial  l(x)}  {\mathbb L}_\alpha({\bf l}, {\bf p}, \lambda, \tau)  |_{ {\bf l}={\bf l}^\dagger, {\bf \lambda}={\bf \lambda}^\dagger, t=t^\dagger, \tau=\tau^\dagger}      &=\mu(x)  -\tau^\dagger D^{-l^\dagger(x)}\log_{e}D+\lambda^\dagger(x) - \sigma^\dagger(x)=0, \forall x\in {\Sigma}, \label{diff1} 
\end{align}
which after manipulation, it becomes
\begin{align}\label{li}
D^{-l^\dagger(x)}=\frac{\mu(x)+\lambda^\dagger(x)-\sigma^\dagger(x)}{\tau^\dagger \log_{e}D}, \hst  x \in {\Sigma}.
\end{align}
Differentiating with respect to $t$ and $s$, the following equations are obtained:
\begin{align}
\frac{\partial        }{\partial t}  {\mathbb L}_\alpha({\bf l}, {\bf p}, \lambda, \tau)  |_{ {\bf l}={\bf l}^\dagger, {\bf \lambda}={\bf \lambda}^\dagger, t=t^\dagger, \tau=\tau^\dagger}  &=  \alpha - \sum_{ x \in {\Sigma}}\lambda^\dagger(x) =0 \Rightarrow \sum_{x \in {\Sigma}} \lambda^\dagger(x) =\alpha. \label{diff2a} \\
\frac{\partial        }{\partial s}  {\mathbb L}_\alpha({\bf l}, {\bf p}, \lambda, \tau)  |_{ {\bf l}={\bf l}^\dagger, {\bf \lambda}={\bf \lambda}^\dagger, t=t^\dagger, \tau=\tau^\dagger}  &= -  \alpha + \sum_{ x \in {\Sigma}}\sigma^\dagger(x) =0  \Rightarrow \sum_{x \in {\Sigma}} \sigma^\dagger(x) =\alpha. \label{diff2b}
\end{align}
When $\tau^\dagger=0$,   \eqref{diff1} gives  $\mu(x)=\sigma^\dagger(x)-\lambda^\dagger(x), \forall  x \in {\Sigma}$. Since $\sigma^\dagger(x) = \lambda^\dagger(x)=0$ $\forall  x \in {\Sigma\setminus\Sigma^o \cup \Sigma_o} $, then it is concluded that $\mu(x)=0$. However, $\mu(x) >0$, $\forall  x \in {\Sigma\setminus\Sigma^o \cup \Sigma_o} $, and therefore, necessarily $\tau^\dagger >0$. Next, $\tau^\dagger$ is found by substituting \eqref{li} and \eqref{diff2a} into the Kraft equality to deduce
\begin{align*}
\sum_{x \in {\Sigma}   }D^{-l^\dagger(x)} & = \sum_{x \in  \mathcal{X}}\frac{\mu(x)+\lambda^\dagger(x)-\sigma^\dagger(x)}{\tau^\dagger  \log_{e}D} \\
& =\frac{ \sum_{ x \in {\Sigma} }\mu(x)}{\tau^\dagger\log_{e}D} +  \frac{ \sum_{ x \in {\Sigma} }\lambda^\dagger(x)}{\tau^\dagger \log_{e}D}-  \frac{ \sum_{ x \in {\Sigma} }\sigma^\dagger(x)}{\tau^\dagger \log_{e}D}  = \frac{1}{\tau^\dagger\log_{e}D} =1 .
\end{align*}
Therefore, $\tau^\dagger=\frac{1}{\log_{e}D}$.
\noi Substituting $\tau^\dagger$ into \eqref{li} yields
\bea\label{eq:55}
D^{-l^\dagger(x)}=\mu(x)+\lambda^\dagger(x)-\sigma^\dagger(x),  &   x \in {\Sigma}.
\eea
Let $w^\dagger(x) \triangleq D^{-l^\dagger(x)}$, i.e., the probabilities that correspond to the codeword lengths $l^\dagger(x)$; also, let  $\underline{w} \triangleq D^{-t^\dagger}$ and $\overline{w} \triangleq D^{-s^\dagger}$. From the Karush-Kuhn-Tucker conditions \eqref{lg6} and \eqref{lg7} we deduce the following. For all $x \in {\Sigma\setminus\Sigma^o \cup \Sigma_o}$, $l(x)<t$ and  $l(x)>s$; hence $\lambda^\dagger(x)=0$ and $\sigma^\dagger(x)=0$. For all $x \in {\Sigma_o}$, $l(x)<t$ and  $l(x)=s$; hence $\lambda^\dagger(x)=0$ and $\sigma^\dagger(x)>0$. For all $x \in {\Sigma^o}$, $l(x)=t$ and  $l(x)>s$; hence $\lambda^\dagger(x)>0$ and $\sigma^\dagger(x)=0$. Therefore, we can distinguish \eqref{eq:55} in the following cases:
\begin{align}
D^{-l^\dagger(x)}&=\mu(x), \hst  x \in{\Sigma\setminus\Sigma^o \cup \Sigma_o} , \label{eq:551}  \\
D^{-l^\dagger(x)}&=\mu(x)-\sigma^\dagger(x), \hst x \in \Sigma_o , \label{eq:552} \\
D^{-l^\dagger(x)}&=\mu(x)+\lambda^\dagger(x), \hst x \in \Sigma^o . \label{eq:553}
\end{align}
\noi Substituting $\lambda^\dagger(x)$ into \eqref{diff2a} we have $\sum_{x \in {\Sigma}} \big(D^{-l^\dagger(x)}-\mu(x)\big)=\alpha$, and substituting $w^\dagger(x) \triangleq D^{-l^\dagger(x)}$ we get
\begin{align}\label{water2}
\sum_{x \in {\Sigma}}\big(w^\dagger(x)-\mu(x)\big)=\alpha.
\end{align}
We know that $\lambda^\dagger(x)\neq 0$ only when $l^\dagger(x)=t^\dagger$; otherwise, $w^\dagger(x)=\mu(x)$. Hence, we can see that    $w^\dagger(x)-\mu(x)=(\underline{w}-\mu(x))^+$ and it is positive only when $l^\dagger(x)=t^\dagger$. Hence, equation \eqref{water2} becomes
\begin{align}\label{water3}
\sum_{x \in {\Sigma}} \big(\underline{w}-\mu(x)\big)^+=\alpha,
\end{align}
where $(f)^+ = \max(0,f)$. This is the classical waterfilling equation \cite[Section 9.4]{2006:Cover} and $\underline{w}$ is the water-level chosen, as shown in Figure \ref{waterproofing}.\\
\\ 
\noindent If we also substitute  $\sigma^\ddagger(x)$ into \eqref{diff2a} we have $\sum_{x \in {\Sigma}} \big(\mu(x)-D^{-l^\ddagger(x)}\big)=\alpha$, and substituting $w^\ddagger(x) \triangleq D^{-l^\ddagger(x)}$ we get
\begin{align}\label{water2a}
\sum_{x \in {\Sigma}}\big(\mu(x) - w^\ddagger(x)\big)=\alpha.
\end{align}
Hence, substituting $\overline{w} \triangleq D^{-s}$, equation \eqref{water2a} becomes
\begin{align}\label{water3a}
\sum_{x \in {\Sigma}} \big(\mu(x) - \overline{w}\big)^+=\alpha.
\end{align}

\begin{remark}
Note that it is possible to handle the case for which $\mu(x)=0$ for some $x\in \Sigma$, in exactly the same way. In this case, $x\in \Sigma^o$ and from equation \eqref{eq:553}, it is deduced that $\lambda^\dagger(x)=0$ at $\alpha=0$, and hence $D^{-l^\dagger(x)}=0$. For $\alpha>0$, it is obvious from equation \eqref{eq:553} that $D^{-l^\dagger(x)}=\lambda^\dagger(x)$.
\end{remark}



\subsection{Proof of Lemma \ref{lemma_optLengths}}\label{proof_lemma_optLengths}

\noi By introducing a real-valued  Lagrange multiplier $\lambda$ associated with the constraint the augmented pay-off is defined by
\begin{align} \label{Lagrangian_1}
 {\mathbb L}_\alpha({\bf l},\bm{\mu}, \lambda) \tri \sum_{ x \in {\Sigma\setminus\Sigma^o \cup \Sigma_o}} l(x) \mu(x)+\left( \sum_{ x \in \Sigma^o}  \mu(x)+\alpha \right) l_{\max} &+
\left( \sum_{ x \in \Sigma_o}  \mu(x)-\alpha \right) l_{\min} \nonumber \\
&+  \lambda \left(\sum_{x \in \Sigma} D^{-l(x)}-1\right).
\end{align}
The augmented pay-off is  a convex and  differentiable function  with respect to ${\bf l}$. Denote the real-valued minimization of \eqref{Lagrangian_1} over ${\bf l}, \lambda$ by  ${\bf l}^\dagger$ and $\lambda^\dagger$. By the Karush-Kuhn-Tucker theorem, the following conditions are necessary and sufficient for optimality:
\begin{eqnarray}
\frac{\partial        }{\partial  l(x)}  {\mathbb L}_\alpha({\bf l}, \bm{\mu}, \lambda) |_{ {\bf l}={\bf l}^\dagger, \lambda=\lambda^\dagger} &=&0,  \label{lg1} \\
\sum_{x \in \Sigma}D^{-l^\dagger(x)}  -1\ &\leq& 0,  \label{lg2} \\
\lambda^\dagger \cdot  \left(\sum_{x \in \Sigma}D^{-l^\dagger(x)}  -1\right)          &=&0, \label{lg3} \\
\lambda^\dagger &\geq& 0. \label{lg4}
\end{eqnarray}
Differentiating with respect to ${\bf l}$, when $x \in {\Sigma\setminus\Sigma^o \cup \Sigma_o}$, $x \in \Sigma_o$ and  $x \in \Sigma^o$ the following equations are obtained:
\begin{align}
\frac{\partial   }{\partial  l(x)}  {\mathbb L}_\alpha({\bf l}, \bm{\mu}, \lambda)  \arrowvert_{ {\bf l}={\bf l}^\dagger, \lambda=\lambda^\dagger}         &=\mu(x)-\lambda^\dagger D^{-l^\dagger(x)}\log_{e}D=0, \hst  x \in {\Sigma\setminus\Sigma^o \cup \Sigma_o}  \label{diff1e} \\
\frac{\partial        }{\partial  l(x)}  {\mathbb L}_\alpha({\bf l}, \bm{\mu}, \lambda)  \arrowvert_{ {\bf l}={\bf l}^\dagger, \lambda=\lambda^\dagger}   &=\sum_{ x \in \Sigma_o}  \mu(x)-\alpha -\lambda^\dagger |\Sigma_o| D^{-l^\dagger(x)}\log_{e}D =0, \hst \: x \in\Sigma_o. \label{diff2e} \\
\frac{\partial        }{\partial  l(x)}  {\mathbb L}_\alpha({\bf l}, \bm{\mu}, \lambda)  \arrowvert_{ {\bf l}={\bf l}^\dagger, \lambda=\lambda^\dagger}   &=\sum_{ x \in \Sigma^o}  \mu(x)+\alpha -\lambda^\dagger |\Sigma^o| D^{-l^\dagger(x)}\log_{e}D =0, \hst \: x \in\Sigma^o. \label{diff3e}
\end{align}
When $\lambda^\dagger=0$,   \eqref{diff1e} gives  $\mu(x)=0, \forall  x \in {\Sigma\setminus\Sigma^o \cup \Sigma_o} $. Since $\mu(x) >0$ then necessarily $\lambda^\dagger >0$. Therefore, \eqref{diff1e}, \eqref{diff2e} and \eqref{diff3e} are equivalent to the following identities:
 \begin{align}
D^{-l^\dagger(x)}&=\frac{\mu(x)}{\lambda^\dagger \log_{e}D}, \hst  x \in{\Sigma\setminus\Sigma^o \cup \Sigma_o} , \label{l1}  \\
D^{-l^\dagger(x)}&=\frac{\sum_{ x \in \Sigma_o}  \mu(x)-\alpha }{\lambda^\dagger  |\Sigma_o| \log_{e}D}, \hst x \in \Sigma_o , \label{l2} \\
D^{-l^\dagger(x)}&=\frac{\sum_{ x \in \Sigma^o}  \mu(x)+\alpha}{\lambda^\dagger  |\Sigma^o| \log_{e}D}, \hst x \in \Sigma^o . \label{l3}
\end{align}
Next, $\lambda^\dagger$ is found by substituting \eqref{l1}, \eqref{l2} and \eqref{l3} into the Kraft equality to deduce:
\begin{align*}
\sum_{x \in \Sigma  }D^{-l^\dagger(x)} & = \sum_{x \in  {\Sigma\setminus\Sigma^o \cup \Sigma_o} }D^{-l^\dagger(x)}+\sum_{x \in \Sigma_o}D^{-l^\dagger(x)} +\sum_{x \in  \Sigma^o}D^{-l^\dagger(x)}\\
& = \sum_{x \in  {\Sigma\setminus\Sigma^o \cup \Sigma_o} }\frac{\mu(x)}{\lambda^\dagger \log_{e}D} + \sum_{x \in \Sigma_o}\frac{\sum_{ x \in \Sigma_o}  \mu(x)-\alpha }{\lambda^\dagger  |\Sigma_o| \log_{e}D}+ \sum_{x \in \Sigma^o}\frac{\sum_{ x \in \Sigma^o}  \mu(x)+\alpha}{\lambda^\dagger  |\Sigma^o| \log_{e}D}\\
& = \frac{ \sum_{x \in  {\Sigma\setminus\Sigma^o \cup \Sigma_o} }\mu(x)}{\lambda^\dagger \log_{e}D} +  |\Sigma_o| \frac{\sum_{ x \in \Sigma_o}  \mu(x)-\alpha }{\lambda^\dagger  |\Sigma_o| \log_{e}D}+ |\Sigma^o|\frac{\sum_{ x \in \Sigma^o}  \mu(x)+\alpha}{\lambda^\dagger  |\Sigma^o| \log_{e}D}\\
& = \frac{\sum_{ x \in  {\Sigma\setminus\Sigma^o \cup \Sigma_o} }\mu(x)+\sum_{x \in \Sigma_o}\mu(x)+ \sum_{x \in \Sigma^o}\mu(x)}{\lambda^\dagger\log_{e}D} \\
& = \frac{1}{\lambda^\dagger\log_{e}D} =1 .
\end{align*}
Substituting $\lambda^\dagger$ into\eqref{l1}, \eqref{l2} and \eqref{l3}  yields
\bes
D^{-l^\dagger(x)}= \left\{ \begin{array}{ll}
\mu(x),  &   x \in  {\Sigma\setminus\Sigma^o \cup \Sigma_o}    \\
 \frac{ \sum_{ x \in \Sigma^o}  \mu(x)+\alpha }{|\Sigma^o|} , &  x \in \Sigma^o \\
\frac{ \sum_{ x \in \Sigma_o}  \mu(x)-\alpha }{|\Sigma_o|}, &   x \in \Sigma_o .
\end{array} \right.
\ees
Finally, from the previous expression one obtains \eqref{eq:solution1}.

\subsection{Proof of Lemma \ref{lemma_weights1}}\label{proof_lemma_weights1}

We can show the validity of the statements in Lemma~\ref{lemma_weights1} by considering five cases. More specifically,
\begin{enumerate}
\item[(i)] $x,y \in  \Sigma\setminus\Sigma_o\cup\Sigma^o $: then $ \nu_\alpha(x)=\mu(x) \leq \mu(y) = \nu_\alpha(y)$,~$\forall~\alpha \in [0,1]$;
\item[(ii)] $x,y \in  \Sigma^o $:  $\nu_\alpha(x)=\nu_\alpha(y)=\underline{\nu}_\alpha \triangleq \min_{ x \in \Sigma} \nu_\alpha(x)$;
\item[(iii)] $x,y \in  \Sigma_o $:  $\nu_\alpha(x)=\nu_\alpha(y)=\overline{\nu}_\alpha  \triangleq \max_{ x \in \Sigma} \nu_\alpha(x)$;
\item[(iv)] $x \in \Sigma^o$,  $y \in \Sigma\setminus\Sigma_o\cup\Sigma^o $ (or $x \in \Sigma\setminus\Sigma_o\cup\Sigma^o $,  $y \in \Sigma^o$): consider the case $x \in  \Sigma^o$,  $y \in \Sigma\setminus\Sigma_o\cup\Sigma^o$. Then, by taking derivatives
\begin{align}
\frac{\partial \nu_{\alpha}(y)}{\partial \alpha}&=0, \hst y \in \Sigma\setminus\Sigma_o\cup\Sigma^o,   \label{lemmW:1} \\
\frac{\partial \nu_{\alpha}(x)}{\partial \alpha}&=\frac{1}{|\Sigma^o|}>0, \hst x \in \Sigma^o.  \label{lemmW:2}
\end{align}
\item[(v)] $x \in \Sigma_o$,  $y \in \Sigma\setminus\Sigma_o\cup\Sigma^o $ (or $x \in \Sigma\setminus\Sigma_o\cup\Sigma^o $,  $y \in \Sigma_o$): consider the case $x \in  \Sigma_o$,  $y \in \Sigma\setminus\Sigma_o\cup\Sigma^o$. Then, by taking derivatives
\begin{align}
\frac{\partial \nu_{\alpha}(y)}{\partial \alpha}&=0, \hst y \in \Sigma\setminus\Sigma_o\cup\Sigma^o,   \label{lemmW:3} \\
\frac{\partial \nu_{\alpha}(x)}{\partial \alpha}&=-\frac{1}{|\Sigma^o|}<0, \hst x \in \Sigma_o.  \label{lemmW:4}
\end{align}
\end{enumerate}
According to \eqref{lemmW:1}, \eqref{lemmW:2},  \eqref{lemmW:3},  \eqref{lemmW:4}, for $\alpha=0, \nu_\alpha(y)|_{\alpha =0}=\mu(y)\geq \nu_\alpha(x)|_{\alpha =0}=\nu(x)$. As a function of $\alpha \in [0,1]$, for $y \in \Sigma\setminus\Sigma_o\cup\Sigma^o$ the weight $\nu_{\alpha}(y)$ remains unchanged, for $x \in \Sigma^o$ the weight $\nu_{\alpha}(z)$ increases, and for $z \in  \Sigma_o$ the weight $\nu_{\alpha}(z)$ decreases. Hence, since $\nu_\alpha(\cdot)$ is a continuous function with respect to $\alpha$, at some $\alpha =\alpha^\prime$, $\nu_{\alpha^\prime}(x)=\nu_{\alpha^\prime}(y)=\underline{\nu}_{\alpha^\prime}$. Suppose that for some $\alpha=\alpha^\prime+d \alpha$, $d \alpha>0$,  $\nu_\alpha(x)\neq \nu_\alpha(y)$.  Then, the lowest weight will increase and the largest weight will remain constant as a function of $\alpha \in [0,1]$  according to \eqref{lemmW:2} and \eqref{lemmW:1}, respectively. We follow similar arguments for $\nu_{\alpha^\prime}(x)=\nu_{\alpha^\prime}(z)=\overline{\nu}_{\alpha^\prime}$.

\subsection{Proof of Lemma~\ref{prop1}}\label{proof_prop1}

The validity of the statement  is shown by perfect induction. Without loss of generality and for simplicity of the proof, suppose that $\beta_1 < \gamma_1$.
\begin{align*}
 \mbox{Firstly, for} \hso \beta=\beta_{1}: \hst  \nu_\alpha(x_{|\Sigma|})=\nu_\alpha(x_{|\Sigma|-1}) \leq \nu_\alpha(x_{|\Sigma|-2}) \leq \ldots \leq \nu_\alpha(x_{1}). \hst
\end{align*}
Suppose that, when $\alpha=\beta_1+d\alpha \in [0,1]$, $d\alpha>0$, then $\nu_\alpha(x_{|\Sigma|}) \neq \nu_\alpha(x_{|\Sigma|-1}) $. Then,
\begin{align}
 {\mathbb L}_\alpha({\bf l}, {\bm \mu}) =\Big(\mu(x_{|\Sigma|})+\mu(x_{|\Sigma|-1})+\alpha \Big) l_{\max}+ \Big(\mu(x_{1})-\alpha \Big) l_{\min}+    \sum_{x   \in   \Sigma\setminus\Sigma_o\cup\Sigma^o}\mu(x) l(x) , \nonumber
\end{align}
and the weights will be of the form $\nu_{\alpha}(x) = \mu(x)$ for $x \in  \Sigma\setminus\Sigma_o\cup\Sigma^o$, $\nu_{\alpha}(x) =\mu(x_{1})-\alpha$ for $x\in \Sigma_{o}$ and $\nu_{\alpha}(x) =\mu(x_{|\Sigma|})+\alpha$ for $x\in\Sigma^{o,1}=\Big\{ x \in \{ x_{|\Sigma|-1}, x_{|\Sigma|} \} \Big\}$.
The rate of change of these weights with respect to $\alpha$ is
\begin{align}
\frac{\partial \nu_{\alpha}(x)}{\partial \alpha}&= 0, ~x \in  \Sigma\setminus\Sigma_o\cup\Sigma^o, \\
\frac{\partial  \nu_{\alpha}(y)}{\partial \alpha}&=1>0, ~y \in\Sigma^{o,1} \label{minprob}.
\end{align}
Hence, the largest of the two stays constant, while the smallest would increase and therefore they meet again. This contradicts the  assumption that  $\nu_\alpha(x_{|\Sigma|}) \neq \nu_\alpha(x_{|\Sigma|-1}) $ for $\alpha>\beta_1$. Therefore, $\nu_\alpha(x_{|\Sigma|}) = \nu_\alpha(x_{|\Sigma|-1}), ~\forall \alpha \in [\beta_1,1)$.\\
Similarly, for $\alpha>\alpha_k, ~k \in \{2,\ldots, {|\Sigma|}-1\}$,  suppose  the weights are
\begin{align*}
\nu_\alpha(x_{|\Sigma|})=\nu_\alpha(x_{|\Sigma|-1})= \ldots =\nu_\alpha(x_{|\Sigma|-k_1})=\nu_{\alpha}^\flat.
\end{align*}
Then, the pay-off is written as
\begin{align*}
 {\mathbb L}_\alpha({\bf l}, {\bm \mu}) = \sum_{ x \in {\Sigma\setminus\Sigma^{o} \cup \Sigma_o}} l(x) \mu(x)+\left( \sum_{ x \in \Sigma^{o,k_1}}  \mu(x)+\alpha \right) l_{\max} +
\left( \sum_{ x \in \Sigma_{o,k_2}}  \mu(x)-\alpha \right) l_{\min}
\end{align*}
\noi Hence,
\begin{align}
\frac{\partial \nu_{\alpha}(x)}{\partial \alpha}&=0, \hst x \in  {\Sigma\setminus\Sigma^{o} \cup \Sigma_o}, \hso \alpha \in (\alpha_k, 1), \label{wi_prob}  \\
| \Sigma^{o,k_1}|\frac{\partial \nu_{\alpha}^\dagger}{\partial \alpha}&=1>0,  \hst x \in \Sigma^{o,k_1}, \hso \alpha \in (\alpha_k, 1).
\end{align}
Finally, in the case that $\alpha>\alpha_{k+1}, ~k \in \{2,\ldots, |\Sigma|-2\}$, if any of the weights $\nu_\alpha( x),~x \in \Sigma^{o,k_1}$,  changes differently than another, then, either at least one probability will become smaller than others and give a higher codeword length, or it will increase faster than the others and hence according to \eqref{wi_prob}, it will stay constant to meet the other weights.
Therefore, the change in this new set of probabilities should be the same, and the cardinality of $\Sigma^{o,k_1}$ increases by one, that is, $|\Sigma^{o,k_1}|=\left|k_1+1\right|, ~k_1\in \{ 1,\ldots |\Sigma|-2 \}$.

With similar arguments we prove that weights $\nu_\alpha( x),~x \in \Sigma_{o,k_2}$ change in the same way and the cardinality of $\Sigma_{o,k_2}$ increases by one.

\bibliographystyle{IEEEtran}
\bibliography{bibliografia}

\end{document}